\documentclass[pra,superscriptaddress,twocolumn]{revtex4-1}
\usepackage{amsfonts,amssymb,amsmath,bbm}
\usepackage[]{graphics,graphicx,epsfig}
\usepackage{amsthm}
\usepackage{color}

\usepackage{amsmath,amsthm,amsfonts,amssymb,amscd}
\usepackage[utf8]{inputenc}
\usepackage{indentfirst}
\usepackage{graphicx}
\usepackage[T1]{fontenc}
\usepackage{mathpazo}
\usepackage{changes, mathrsfs}
\usepackage{mathtools}

\usepackage[USenglish]{babel}
\usepackage[colorlinks=true]{hyperref}
\usepackage{hyperref}
\usepackage{bm}
\newcommand{\mean}[1]{\langle#1\rangle}

\def\id{\leavevmode\hbox{\small1\kern-3.8pt\normalsize1}}

\def\identity{\leavevmode\hbox{\small1\kern-3.8pt\normalsize1}}

\newtheorem{mydef}{Definition}

\renewcommand{\epsilon}{\varepsilon}

\newtheorem{definition}{Definition} %[section]
\newtheorem{prop}[definition]{Proposition}

\newtheorem{lemma}[definition]{Lemma}

\newtheorem*{thm}{Theorem}

\newtheorem*{rep@theorem}{\rep@title}
\newcommand{\newreptheorem}[2]{%
\newenvironment{rep#1}[1]{%
 \def\rep@title{#2 \ref{##1} (restatement)}%
 \begin{rep@theorem}}%
 {\end{rep@theorem}}}
\makeatother

\newreptheorem{lem}{Lemma}

\def\ba#1\ea{\begin{align}#1\end{align}}
\def\ban#1\ean{\begin{align*}#1\end{align*}}

\newcommand{\ot}{\otimes}
\newcommand{\be}{\begin{equation}}
\newcommand{\ee}{\end{equation}}

\def\benum{\begin{enumerate}}
\def\eenum{\end{enumerate}}

\newcommand{\ket}[1]{|#1\rangle}

\newcommand{\<}{\langle}
\renewcommand{\>}{\rangle}

\def\id{{\operatorname{id}}}

\def\be{\begin{equation}}
\def\ee{\end{equation}}
\def\ben{\begin{eqnarray}}
\def\een{\end{eqnarray}}
\def\ot{\otimes}

\def\bei{\begin{itemize}}
\def\eei{\end{itemize}}

\mathchardef\ordinarycolon\mathcode`\:
\mathcode`\:=\string"8000
\def\vcentcolon{\mathrel{\mathop\ordinarycolon}}
\begingroup \catcode`\:=\active
  \lowercase{\endgroup
  \let :\vcentcolon
  }

\newcommand{\nc}{\newcommand}
 \nc{\proj}[1]{|#1\rangle\!\langle #1 |} 
\nc{\avg}[1]{\langle#1\rangle}

\nc{\conv}{\operatorname{conv}}
\nc{\smfrac}[2]{\mbox{$\frac{#1}{#2}$}} \nc{\Tr}{\operatorname{Tr}}
\nc{\ox}{\otimes} \nc{\dg}{\dagger} \nc{\dn}{\downarrow}
\nc{\lmax}{\lambda_{\text{max}}}
\nc{\lmin}{\lambda_{\text{min}}}

\nc{\csupp}{{\operatorname{csupp}}}
\nc{\qsupp}{{\operatorname{qsupp}}} \nc{\var}{\operatorname{var}}
\nc{\rar}{\rightarrow} \nc{\lrar}{\longrightarrow}
\nc{\poly}{\operatorname{poly}}
\nc{\polylog}{\operatorname{polylog}} \nc{\Lip}{\operatorname{Lip}}
\nc{\Om}{\Omega}
\nc{\wt}[1]{\widetilde{#1}}

\def\>{\rangle}
\def\<{\langle}

\nc{\glneq}{{\raisebox{0.6ex}{$>$}  \hspace*{-1.8ex} \raisebox{-0.6ex}{$<$}}}
\nc{\gleq}{{\raisebox{0.6ex}{$\geq$}\hspace*{-1.8ex} \raisebox{-0.6ex}{$\leq$}}}

\nc{\vholder}[1]{\rule{0pt}{#1}}
\nc{\wh}[1]{\widehat{#1}}
\nc{\h}[1]{\widehat{#1}}

\nc{\ob}[1]{#1}

\def\beq{\begin {equation}}
\def\eeq{\end {equation}}

\def\be{\begin{equation}}
\def\ee{\end{equation}}

\nc{\eq}[1]{(\ref{eq:#1})} 
\nc{\eqs}[2]{\eq{#1} and \eq{#2}}

\nc{\eqn}[1]{Eq.~(\ref{eqn:#1})}
\nc{\eqns}[2]{Eqs.~(\ref{eqn:#1}) and (\ref{eqn:#2})}

\newcommand{\Xs}{x}
\newcommand{\mi}{i}
\newcommand{\mj}{j}
\newcommand{\mk}{k}
\newcommand{\ml}{l}
\newcommand{\nvs}{O}

\nc{\region}{\cS\cW}
\newenvironment{protocol*}[1]
  {
    \begin{center}
      \hrulefill\\
      \textbf{#1}
  }
  {
    \vspace{-1\baselineskip}
    \hrulefill
    \end{center}
  }

\begin{document}
\title{On the tightness of correlation inequalities with no quantum violation}
%\title{Faces in the set of quantum correlations}
\author{Ravishankar \surname{Ramanathan}}
\email{ravishankar.r.10@gmail.com}
\affiliation{National Quantum Information Center of Gda\'nsk,  81-824 Sopot, Poland}
\affiliation{Institute of Theoretical Physics and Astrophysics, University of Gda\'{n}sk, 80-952 Gda\'{n}sk, Poland}
\author{Marco T\'ulio Quintino}
\affiliation{D\'epartement de Physique Th\'eorique, Universit\'e de Gen\`eve, 1211 Gen\`eve, Switzerland}
\author{Ana Bel\'en Sainz}
\affiliation{H. H. Wills Physics Laboratory, University of Bristol, Tyndall Avenue, Bristol, BS8 1TL, United Kingdom}
\author{Gl\'aucia Murta}
%%\affiliation{National Quantum Information Center of Gda\'nsk, 81-824 Sopot, Poland}
\affiliation{Departamento de Fisica, Universidade Federal de Minas Gerais, Caixa Postal 702, 30123-970, Belo Horizonte, MG, Brazil}
\affiliation{QuTech, Delft University of Technology, Lorentzweg 1, 2628 CJ Delft, The Netherlands}
%\affiliation{University of Gda\'nsk, 80-952 Gda\'nsk, Poland}
\author{Remigiusz Augusiak}
\affiliation{Center for Theoretical Physics, Polish Academy of Sciences, Aleja Lotnik\'ow 32/46, 02-668 Warsaw, Poland}

\begin{abstract}

We study the faces of the set of quantum correlations, i.e., the Bell and noncontextuality inequalities without any quantum violation.  
First, we investigate the question whether every proper (tight) Bell inequality for two parties, other than the trivial ones from positivity, normalization and no-signaling can be violated by quantum correlations, i.e., whether the classical Bell polytope or the smaller correlation polytope share any facets with their respective quantum sets. To do this, we develop {a recently derived bound on the quantum value of linear games} based on the norms of game matrices to give a simple sufficient condition to identify linear games with no quantum advantage. {Additionally we show how this bound can be extended to the general class of unique games.} We then show that the paradigmatic examples of correlation Bell inequalities with no quantum violation, namely the non-local computation games do not constitute tight Bell inequalities, not even for the correlation polytope. We also extend this to an arbitrary prime number of outcomes for a specific class of these games. 
We then study the faces in the simplest CHSH Bell scenario of binary dichotomic measurements, and identify edges in the set of quantum correlations in this scenario.

Finally, we relate the non-contextual polytope of single-party correlation inequalities with the cut polytope $CUT(\nabla G)$, where $G$ denotes the compatibility graph of observables in the contextuality scenario and $\nabla G$ denotes the suspension graph of $G$. We observe that there exist tight non-contextuality inequalities with no quantum violation, and furthermore that this set of inequalities is beyond those implied by the Consistent Exclusivity principle.
\end{abstract}

\maketitle
\section{Introduction}

In recent years there has been a growing interest in non-local quantum correlations, i.e., 
correlations generated in a Bell-type experiment which violate Bell inequalities \cite{Bell}.
On the one hand, their existence shows that quantum and classical mechanics 
fundamentally depart one from another, and, on the other hand, 
non-locality has been turned into a powerful resource allowing 
to perform tasks not known in classical physics. These are, for instance,
generation of cryptographic key that is secure against 
even post-quantum eavesdroppers \cite{securekey}, reduction of communication complexity \cite{CommCompl}, or
true randomness certification and amplification \cite{rand}.
{Therefore, for fundamental reasons as well as for applications,}  detection and characterization of the set of non-local correlations is of utmost importance in quantum information theory. 
{This not only restricts to correlations in Bell scenarios but also applies to general contextuality ones \cite{KS,AB,AFLS}, 
where contextual behaviours (i.e. conditional probability distributions stronger than allowed {by deterministic noncontextual models})
prove to be a necessary resource for quantum computation \cite{contcomp}.}

{In a bipartite Bell scenario, }we consider correlations between two parties Alice and Bob, who can perform $m_A, m_B$ measurements, and obtain $d_A, d_B$ outcomes respectively. A box $\mathcal{P}$ describes the set of conditional probability distributions $P(a,b|x,y)$ that can be obtained in such a Bell experiment, $\mathcal{P}$ is usually expressed as a real vector in $\mathbb{R}^{m_A m_B d_A d_B}$. The constraints of normalization, no-signaling and non-negativity of probabilities {imposed on $\mathcal{P}$ define the set of no-signaling boxes which constitutes a polytope in a reduced dimension. } Contained within the no-signaling polytope is the set of interest, namely that of the quantum correlations obtainable {from measurements in a bipartite quantum state,} and inside this set is the classical polytope of correlations realizable in a local realistic theory. For most of the applications mentioned previously and for fundamental reasons, we are interested in studying how the set of quantum correlations fits in between the classical and no-signaling polytopes. 

A proper (tight) Bell inequality is a facet of the classical polytope, which is not also a facet of the no-signaling one, i.e., one that does not merely rephrase the no-signaling, positivity and normalization conditions. In \cite{Gill}, Gill posed as an open question whether every proper (tight) Bell inequality is violated by quantum correlations. As a variant of the question, Avis et al. \cite{AII06} ask whether the correlation polytope in the binary outcome scenario (the set of correlations $\langle A_i B_j \rangle$ observable in a two-party experiment, not including the local terms $\langle A_i \rangle$, $\langle B_j \rangle$ themselves) shares any facets in common with the set of quantum correlations. In \cite{GYNI} it was shown that if three or more parties are involved in the Bell experiment, there exist multipartite tight Bell inequalities that do not allow for quantum violation. {These inqualities} led to the formulation of an information-theoretic principle called local orthogonality {(LO) \cite{LO}}, which serves to bound the set of correlations realizable in a physical theory. 
{In the bipartite case  LO reduces to the non-signaling condition, therefore bringing no nontrivial constraints to the set of bipartite correlations, and 
 all of the tight bipartite Bell inequalities found so far in the literature (as listed in \cite{facets}) are known to be violated in quantum theory. }
 {In the multipartite scenario it remains an open problem whether all facet-defining nontrivial Bell inequalities are of the local orthogonality form.}
{Let us remark that, beyond the scope of this paper, other questions on how the quantum set fits between the classical and no-signaling ones have also been intensively studied; in particular unbounded violations of Bell inequalities, i.e., large separations between the classical and quantum sets \cite{unb-viol}.} 

{The questions of how the quantum set of behaviours fits between the classical one and those explained by more general theories is one main research topic as well in contextuality \cite{KS}. Contextuality is a curious phenomenon that also signals the 
nonclassicality of quantum theory, and may be understood as a more general case of nonlocality 
scenarios that includes measurements on a single system \cite{AB, AFLS} 
 (there exist other inequivalent approaches to contextuality that will not be discussed in this work like \textit{preparation and measurement contextuality by
Spekkens \textit{et. al.} 
\cite{spekkens})}. Different aspects of the phenomenon have been explored in the literature \cite{adanrays,AQB+13,horodecki_measure}, 
and special interest has been put into understanding the boundary between quantum and more general theories as well. One example of 
this is the extension of LO to contextuality scenarios, called Consistent Exclusivity principle, that serves to bound
(although not tightly) the set of quantum behaviours. In particular, there is the similar open question of whether every facet of the 
classical polytope that coincides with one of the quantum set is of the CE form. }

In this paper, we first study the question of whether there are tight two-party Bell inequalities with no quantum violation. In addition to the binary outcome correlation Bell inequalities explored in \cite{AII06}, we also study their natural generalization to more outcomes known as linear games. We develop a simple sufficient condition for these games to 
exhibit no quantum violation based on the singular vectors of their game matrices \cite{RAM}, with which we identify new inequalities with this property.  A well-known class of two-party inequalities that do not allow for quantum violation are those arising from another information-theoretic principle, called no-NLC (no advantage in non-local computation) \cite{NLC}. 
Its generalization to the case of more outcomes was considered in \cite{RAM}. We show by an explicit decomposition into other valid inequalities 
of the classical and the correlation polytopes that a broad set of these inequalities do not constitute facets of both these polytopes. 
{Finally, we
also study the question {of tight inequalities} in the single-party scenario, namely non-contextuality inequalities with binary outcomes. 
{On the one hand, we identify the non-contextual polytope  that arises from the compatibility graph representing the measurements in the
experiment with a well explored object in  computer science named \textit{cut polytope} \cite{DL97}. We use this relation to construct
many tight non-contextuality inequalities with no quantum violation.} On the other hand, we further characterise the polytope of behaviours 
that satisfy CE in this example, and show that not all tight non-contextuality inequalities with no quantum violation are among its facets.}

\section{Preliminaries}
\label{Preliminaries}

In this section, we establish the preliminary notions concerning the classical and no-signaling polytopes, and the notion of tight Bell inequalities. 
{Then we present the basics of contextuality scenarios and their sets of behaviours.}

\paragraph{Polytopes.} %Let us begin with polytopes (see Refs. \cite{Schrijver,comp-geom}) and 
Consider
a linear space $\mathbb{R}^n$ whose elements we denote by $\mathbf{p}$.
A convex polytope $P \subseteq \mathbb{R}^n$ (see \cite{Schrijver,comp-geom}) is the convex hull of a finite number of points in $\mathbb{R}^n$.
Alternatively, $P$ is a bounded polyhedron, where $Q\subset \mathbb{R}^n$ is a polyhedron if there is a 
system of finitely many inequalities $\textbf{C} \cdot \textbf{q} \leq \textbf{b}$ such that 
$Q = \{ \textbf{q} | \textbf{C} \cdot \textbf{q} \leq \textbf{b}\}$. Here, $\mathbf{C}$ is a $m\times n$ 
matrix and $\mathbf{b}\in\mathbb{R}^m$ with $m$ being the number of inequalities.
   
Let $P \subseteq \mathbb{R}^{n}$, $\textbf{c} \in \mathbb{R}^n$ and $b \in \mathbb{R}$. An inequality $\textbf{c} \cdot \textbf{p} \leq b$ is valid for $P$ if it holds for all $\textbf{p} \in P$.  The hyperplane $H_{(\textbf{c}, b)}$ given as
\begin{equation}
H_{(\textbf{c}, b)} = \{ \textbf{p} \in \mathbb{R}^n | \textbf{c} \cdot \textbf{p} = b\}
\end{equation} 
is said to be a supporting hyperplane of the polytope $P$ if $\max\{ \textbf{c} \cdot \textbf{p} | \textbf{p} \in P\} = b$. Then, 
$F$ is a face of $P$ if $F = P$ or $F = P \cap H$ for some supporting hyperplane $H$ of $P$. A set $Q = \{\textbf{p}_1, \dots, \textbf{p}_k\} \subseteq \mathbb{R}^n$ is affinely dependent iff there exist $\lambda_1, \dots, \lambda_k \in \mathbb{R}$ not all zero such that $\sum_{i=1}^{k} \lambda_i \textbf{p}_i = \textbf{0}$ and $\sum_{i=1}^{k} \lambda_i = 1$. The dimension of $P$ is the smallest dimension of its affine hull, i.e., $dim(P) := \max\{|Q| | Q \subseteq P, \text{Q affinely independent}\}$. $F$ is a facet of $P$ if $F$ is a face and $dim(F) = dim(P)  - 1$. An inequality $\textbf{c} \cdot \textbf{p} \leq b$ is said to be facet-defining (or essential) for $P$ if $H_{(\textbf{c}, b)}\cap P $ is a facet of $P$. Let $\textbf{C} \cdot \textbf{p} \leq \textbf{b}$ be a system of valid inequalities for polytope $P$ such that for each facet $F$ of $P$, there is a row $\textbf{c}_i \cdot \textbf{p} \leq \textbf{b}_i$ of $\textbf{C} \cdot \textbf{p} \leq \textbf{b}$ such that $F = P \cap H_{\textbf{c}_i, \textbf{b}_i}$. Then $P = \{ \textbf{p} \in \mathbb{R}^n | \textbf{C} \cdot \textbf{p} \leq \textbf{b}\}$.

\paragraph{Bell scenario and sets of correlations.} 
As explained in the introduction, we consider correlations between two parties Alice and Bob, who can perform $m_A, m_B$ measurements, and obtain $d_A, d_B$ outcomes respectively. Such a scenario is usually referred to as $(2,m_A,m_B,d_A,d_B)$. The correlations the parties 
generate in this way is described by the collection 
$\mathcal{P} = \{P(a,b|x,y)\}$ of conditional probabilities $P(a,b|x,y)$ that 
Alice and Bob obtained $a$ and $b$ upon performing the $x$-th and
$y$-th measurement, respectively with $a (b)=0,\ldots,d_A-1 (d_B - 1)$ and $x (y)=1,\ldots,m_A (m_B)$.
The usual way of dealing with such a collection is to treat it as a vector  $\textbf{P}$ from $\mathbb{R}^{m_A m_B d_A d_B}$.
Now, in addition to the normalization conditions $\sum_{a,b} P(a,b|x,y) = 1$ for all $x,y$, the no-signaling conditions
\begin{equation}\label{NS1}
\sum_{a}P(a,b|x,y)=\sum_{a}P(a,b|x',y)
\end{equation}
for any $b$, $y$ and any pair $x\neq x'$, and
\begin{equation}\label{NS2}
\sum_{b}P(a,b|x,y)=\sum_{b}P(a,b|x,y')
\end{equation}
for any $a$, $x$, and any pair $y\neq y'$, are imposed.  The set of boxes $\mathcal{P}$ {satisfying normalization and no-signaling} thus forms a convex polytope known as the no-signaling polytope $\mathcal{N}$ of dimension $D = m_A m_B (d_A - 1) (d_B - 1) + m_A (d_A -1) + m_B (d_B - 1)$.

In quantum theory Alice and Bob share some quantum state $\rho$ acting on some  Hilbert space $\mathcal{H}_A\ot\mathcal{H}_B$, and perform POVM measurements on their share of this state. The set $\mathcal{Q}$ of quantum boxes so obtained forms a convex set if one does not constrain the dimension of $\rho$. 
In general, however, 
the boundary of $\mathcal{Q}$ remains unknown and is very difficult to determine
(see nevertheless, \cite{Masanes} for results in the two-input two-output scenario). Interestingly, some of the elements of $\mathcal{Q}$ 
can be generated by the parties without using any quantum state; in fact, the only 
resource needed to obtain them is some shared classical information,
also referred to as \textit{shared randomness},
represented by a random variable $\lambda$ with probability distribution $p(\lambda)$.
Such correlations are those for which 
\begin{equation}\label{local}
P(a,b|x,y)=\sum_{\lambda}p(\lambda)P_A(a|x,\lambda)P_B(b|y,\lambda).
\end{equation}
and are said to admit a \textit{local hidden variable} model. For simplicity we also call them \textit{local} or \textit{classical} {correlations}.
Here $P_A(a|x,\lambda)$ and $P_B(b|y,\lambda)$ are local probability distributions.

Local correlations again form a convex set, this time, however, 
it is a polytope $\mathcal{L}$, whose vertices 
are local deterministic correlations $\{P_{\mathrm{det}}(a,b|x,y)\}$ 
for which $P_{\mathrm{det}}(a,b|x,y)=P_A(a|x)P_B(b|y)$ 
and all local probabilities $P_A(a|x), P_B(b|y)\in\{0,1\}$ with $x (y)=1,\ldots,m_A (m_B)$
and $a (b)=1,\ldots,d_A (d_B)$. 
Since the work of Bell \cite{Bell}, we know that $\mathcal{L} \subsetneq \mathcal{Q}$, i.e., 
there exist quantum correlations $\mathbf{p}\in \mathcal{Q}$, termed non-local, 
which cannot be written as in Eq. (\ref{local}). 
%In order to do so, 
%he introduced the concept of Bell inequalities which  
And subsequently, from the work  of Popescu and Rohrlich \cite{PR} we have that the chain of inclusions
$\mathcal{L} \subsetneq \mathcal{Q}\subsetneq \mathcal{N}$ holds true. That is, 
local correlations form a proper subset of the set of quantum correlations, while the latter forms 
a proper subset of all nonsignaling correlations. There exist quantum correlations
that are not local, usually termed non-local, as well as nonsignaling correlations which are not quantum.

The Bell inequalities are linear inequalities 
\begin{equation}\label{dupa}
\mathbf{c}\cdot\mathbf{P}\leq b
\end{equation}
that constrain the local set $\mathcal{L}$, and their violation 
signals non-locality. Here, $\mathbf{c}\in\mathbb{R}^N$ is some constant vector and 
$b$ stands for the so-called classical bound defined as
$b=\max_{\mathbf{p}\in P_{\mathrm{L}}}\mathbf{c}\cdot\mathbf{P}$. 
As explained above, owing to the fact that $\mathcal{L}$ is a polytope, a finite amount of
such Bell inequalities is sufficient to fully characterize it. These correspond
to the facets of $\mathcal{L}$ and are usually called
\textit{tight Bell inequalities} for the local Bell polytope. 
To check then whether a given Bell inequality (\ref{dupa})
defines a facet of the corresponding local polytope of dimension $D$
one needs to show that the classical deterministic boxes 
$\{P_{\mathrm{det}}(a,b |x,y)\}$ that achieve value $b$ 
for the inequality span an affine subspace of dimension $D-1$.

In the study of tight two-party Bell inequalities, another local polytope has also been of interest namely the correlation polytope in the binary 
outcome scenario $(2,m_A,m_B,2,2)$. This is the polytope of possible two-party correlations $\langle A_x B_y \rangle$ achievable in a local realistic 
theory where $A_x, B_y$ denote dichotomic {events and $\langle A_x B_y \rangle=P(a=b|x,y)-P(a\neq b|x,y)$}. While the local Bell polytope in this scenario 
is of dimension $m_A m_B + m_A + m_B$, the correlation polytope is of smaller dimension $m_A m_B$ by virtue of not considering the local
terms $\langle A_x \rangle$ and $\langle B_y \rangle$. In this binary output situation, there is a rich literature relating the local Bell polytope, 
the correlation polytope and the set of quantum correlations to convex sets studied in polyhedral combinatorics \cite{AII06, AHW10}.

\paragraph{{Contextuality scenarios and sets of behaviours.}} 

	Here we will focus on contextuality scenarios arising from compatibility graph scenarios, like the one discussed by Klyachko \cite{klyachko1,klyachko2} and equivalent to the marginal problem \cite{planeta,fritz}.
	
	{There are two ways to approach contextuality scenarios, both starting from a hypergraph but giving the vertices and edges a different interpretation. One starts from a compatibility hypergraph, where its vertices represent measurements and the hyperedges the sets of compatible measurements.  The other starts from an events  hypergraph, where vertices represent measurement outcomes and the hyperedges the measurements   \cite{AFLS}. Both include Bell scenarios as a particular family of general contextuality ones. Here we will briefly review the concepts relevant to this work.}

{Within the hypergraph approach of \cite{AFLS}, a contextuality scenario is defined as a hypergraph $H=(V,E)$ whose vertices $v \in V$ correspond to the events of the scenario, and where the hyperedges $e = \{v_1 , \cdots , v_k \} \in E$ (subsets of $V$) are the measurements of the scenario, with the vertices as the allowed measurement outcomes. In addition, every measurement is assumed to be complete, in the sense that every behaviour over the contextuality scenario $H$ satisfies the normalization condition $\sum_{v \in e} P(v) = 1$ for every $e$, where $P(v)$ denotes the probability that outcome $v$ is obtained given that the measurement $e$ is performed. As classical behaviours we consider the ones that are explained via deterministic non-contextual (NC) hidden variable models and convex combinations of them; i.e. a behavior is classical if it can be explained by a convex sum of the ones where only one outcome happens with certainty for each measurement. The set of all classical behaviours forms the non-contextual polytope. Its facets are referred to as tight noncontextuality inequalities.}

{On a similar footing, a behaviour is quantum whenever the probabilities arise as $P(v) =\mathrm{tr}(\rho P_v)$, where $\rho$ and $\{P_v\}_{v \in V}$ are hermitian operators over some Hilbert space $\mathcal{H}$, and the projectors $P_v$ satisfy the following conditions: $\sum_{v \in e} P_v = \openone_{\mathcal{H}}$ for every measurement $e$ and $P_v \perp P_u$ when $v$ and $u$ belong to the same measurement. }

{Within this framework, a natural notion of exclusiveness among events arises: two events which are outcomes of the same measurement are naturally exclusive. Then, in the language of graph theory, two distinct vertices $u$ and $v$ are orthogonal (denoted by $u \perp v$) if there exists a hyperedge $e \in E$ such that $u \in e$ and $v \in e$. A proposed principle to bound the set of quantum models for contextuality scenarios is {\it Consistent Exclusivity} (CE) \cite{Henson, AFLS}, also refered to as {\it Global Exclusivity} \cite{cabello}. Here, we focus only on the constraints that it imposes over one copy of the system, which are expresed as follows: 
\begin{definition}\label{defCE}
A probabilistic model $P$ on a contextuality scenario $H$ satisfies CE$^1$ when $\sum_{v \in S} P(v) \leq 1$ for every set $S \subset V$ of mutually orthogonal events. 
\end{definition} 
We refer to the conditions in Def. \ref{defCE} as CE$^1$ inequalities, and usually denote each them by the set of orthogonal events $S$ that gives rise to it. It can be easily verified that, just like the Local Ortogonality inequalities \cite{LO}, the CE$^1$ inequalities also cannot be violated in quantum theory \cite{AFLS}.

{The other successful approach to contextuality \cite{AB} focuses instead on the set of measurements that are performed on the (possibly multipartite) system and the compatibility relations among them. Within this language, a compatibility hypergraph  is one where its vertices represent measurable quantities (namely, measurements) and hyperedges represent quantities that can be jointly measured i.e., compatible measurements. In addition, the hypergraph can be equipped with vertex weights denoting the number of possible outcomes of the corresponding measurement. Here we will focus on dichotomic observables, and hence simplify the scenario to unweighted hypergraphs.}

{Within this viewpoint, an assignment of probabilities to measurement outcomes (denoted as well by behaviour) is well defined as long as it satisfies the `no disturbance' condition (also referred to as No-Signaling in Bell scenarios, or sheaf-condition in \cite{AB}); that is, the (marginal) probability for obtaining an outcome when performing a measurement should not depend on the choice of other compatible measurements that are performed alongside.}

{Given a contextuality scenario with compatibility graph $G(V,E)$, a noncontextual deterministic behaviour is an assignment of an outcome to each measurement $M_i$ which does not depend on the other measurements that are compatible with it. This assignment then tells what are the outcomes that occur with certainty whenever compatible measurements are performed. When the scenario consists of dichotomic measurements, such deterministic assignments are equivalently determined by the assignment of a value $\pm 1$ to the {\it single correlators} $\mean{M_i} = P(1\vert\mi) - P(-1 \vert\mi)$. Moreover, in the case where the sets of compatible measurements have two elements (i.e. the compatibility structure is given by  a graph), the object of interest is the conditional probability distribution (a.k.a behaviour) $p(ab|ij)$, where $M_i$ and $M_j$ are compatible. A behaviour that satisfies the no-disturbance principle is equivalently represented by the single correlators $\mean{M_i}$ and its \textit{full correlators} $\mean{M_i M_j} := P(a = b \vert\mi , \mj) - P(a\neq b \vert\mi , \mj)$ \cite{GisinCG, AQB+13}.  On the one hand, the non-contextual (NC) polytope of behaviours (for a compatibility graph $G(V,E)$) is defined as the convex hull of deterministic noncontextual ones in $G(V,E)$. On the other hand, another object of interest is the \textit{full-correlation} non-contextual polytope (FC-NC). This is indeed the projection of the NC polytope onto the subspace defined by the full correlators $\mean{M_i M_j}$. Facets of the full-correlation polytope are also facets of the complete one, and a violation of a full-correlation inequality implies a violation of a standard non-contextual one.}

\quad

The interest in finding facet-defining inequalities with no quantum advantage can simply be stated as an interest in finding the largest dimensional 
face of the set of quantum correlations that one can describe analytically. Such faces contribute to an enhanced understanding of this set that one may 
then use to identify information-theoretic principles underlying quantum correlations {for both Bell and general contextuality scenarios}.

\section{Identifying inequalities with no quantum advantage}
\label{SECIII}
In this section, we study the question of how one may identify inequalities with no quantum advantage. For the binary outcome correlation inequalities, 
the quantum violation of the inequality is known to be calculable by a semidefinite program \cite{Tsirelson,Wehner}\cite{footnote1}.
{In Ref.~\cite{RKM+14} we proposed a necessary and sufficient condition for the lack of quantum advantage for these inequalities (XOR games).}
A simpler sufficient but not necessary condition was also given, namely that the game matrices have maximum singular vectors with $\pm 1$ entries 
only. Here, we extend the condition {of no quantum advantage} to the many-outcome scenario of linear games \cite{RAM, Hastad},
which are a class of the well-known unique games \cite{Khot}. 

\textit{A bound on the quantum value of three outcome unique games.}
Unique games are a generalization of \textsc{xor} games to arbitrary output alphabet and are defined as follows:
\begin{mydef}
A two-player unique game $(\textsl{g}^{u}, q)$ is one where two players Alice and Bob receive questions $x$, $y$ from sets $Q_A$ and $Q_B$ respectively, chosen from a probability distribution $q(x,y)$ by a referee. They reply with respective answers $a, b \in [d]$. The game is defined by a winning constraint $b = \pi_{(x,y)}(a)$ for some set of permutations $\{\pi_{(x,y)} \} \subset \mathcal{S}_d$, where $\mathcal{S}_d$ denotes the permutation group on $d$ elements.  
\end{mydef}
A sub-class of the unique games are the \textsc{linear} games, where the output alphabet $[d]$ is identified with an Abelian group $G$ of size $d$, i.e., $a, b \in (G, +)$; and the winning constraint is given by $a + b = f(x,y)$ where $+$ denotes the addition operation in the Abelian group $G$ and $f : Q_A \times Q_B \rightarrow G$. The special case of the cyclic groups $\mathbb{Z}_d$ (integers under addition modulo $d$) is called an \textsc{xor}-d game \cite{RAM}, and the further restriction to $\mathbb{Z}_2$ defines the \textsc{xor} game \cite{Tsirelson}. The class of \textsc{xor} games forms a highly interesting class of Bell inequalities that are also called ``correlation Bell inequalities" where the correlation function $\mathit{E}_{x,y} = \sum_{k=0,1} (-1)^k P(a \oplus b \; \text{mod 2} = k|x,y)$.

The value of the unique game is given by the expression 
\begin{equation}
\omega(\textsl{g}^{u}) = \sum_{\substack{x \in Q_A \\ y \in Q_B}} \sum_{a,b \in G} q(x,y) V(a,b|x,y) P(a,b|x,y),
\end{equation}
where $V(a,b|x,y) = 1$ if $b =\pi_{(x,y)}(a)$ and $0$ otherwise. 
The maximum classical value of the game (the maximum over all deterministic assignments of $a, b$ or their convex combinations) is denoted $\omega_c(\textsl{g}^{u})$, the value of the game achieved by a quantum strategy (POVM measurements on a shared entangled state of arbitrary Hilbert space dimension) is denoted $\omega_q(\textsl{g}^{u})$, while the value achieved by no-signaling strategies (where neither party can signal their choice of input using the correlations) is denoted $\omega_{ns}(\textsl{g}^{u})$. These games have been studied \cite{Hastad, Khot} in the context of hardness of approximation of several important optimization problems, in attempts to identify the existence of polynomial time algorithms to approximate the optimum solution of the problem to within a constant factor. For every unique game, $\omega_{ns}(\textsl{g}^{u}) = 1$ since a no-signaling box exists that wins the game. Such a box is defined by the entries $P(a,b|x,y) = 1/d$ if $b = \pi_{(x,y)}(a)$ and $0$ otherwise for all input pairs $(x,y)$, this strategy clearly wins the game, and is no-signaling since the output distribution seen by each party is fully random for every input, i.e., $P(a|x) = P(b|y) = 1/d$.

\subsection{Bounds on quantum value for unique games}

Let us first bound the quantum value of a unique game with three outputs using generalized norms from a set of matrices, bounds for more outcome unique games and their applications will be discussed elsewhere. 
\begin{prop}\label{thm2}
\label{norm-bound}
The quantum value of a three output unique game $\textsl{g}^u$ with input sets $ Q_A, Q_B$  can be bounded as 
\begin{equation}
\omega_q(\textsl{g}^{u}) \leq \frac{1}{3} \left[ 1 + \sqrt{m_Am_B} \sum_{k=1,2} \| \Phi^{H}_k, \Phi^{(01)H}_k \|_{\text{gen}} \right].
\end{equation}
%\left[1 + 3\| \Phi^{H}_1, \Phi^{(01)H}_1 \|_{\text{gen}} + 3 \| \Phi^{H}_2, \Phi^{(01)H}_2 \|_{\text{gen}} \right].
where 
\begin{equation}\label{gamematrix1}
\Phi^{H}_{k} = \sum_{(x,y) \in Q_A \times Q_B} q^H(x,y) \zeta^{-k f(x,y)} | x \rangle \langle y|
\end{equation}
and
\begin{equation}\label{gamematrix2}
\Phi^{(01) H}_{k} = \sum_{(x,y) \in Q_A \times Q_B} q^{(01) H}(x,y) \zeta^{k f(x,y)} |x \rangle \langle y|
\end{equation}
are the game matrices, with $\zeta= \exp(2 \pi i/3)$, $q^H(x,y)$ and $q^{(01) H}(x,y)$ are probability distributions as specified in the proof bellow, and %while
$\| \Phi^{H}_k, \Phi^{(01)H}_k \|_{\text{gen}}$
is defined as
\begin{eqnarray}
\label{eq:gen-norm}
&&\| \Phi^{H}_k, \Phi^{(01)H}_k \|_{\text{gen}} := \nonumber \\
&& \; \; \max \{ \| \Phi^{H}_k x_1 +  \Phi^{(01)H}_k x_2 \| : \| x_1\| = 1, \|x_2 \|=1 \}.\nonumber\\
\end{eqnarray}
\end{prop}

\begin{proof}
Consider a quantum strategy given by projective measurements $\{ \Pi_{x}^{a} \}$ for Alice and $\{ \Sigma_{y}^{b} \}$ for Bob performed on a pure quantum state $ | \Psi \rangle \in \mathbb{C}^{d} \otimes \mathbb{C}^d$ for some arbitrary dimension $d$. Let us introduce the generalized correlators $\langle A_{x}^{k} \otimes B_{y}^{l} \rangle$ defined via the Fourier transform of the probabilities $P(a,b|x,y)$ as
\begin{equation}
\langle A_{x}^{k} \otimes B_{y}^{l} \rangle = \sum_{a, b = 0}^{2} \zeta^{-(a k + b l)} P(a,b|x,y),
\end{equation}  
where $\zeta = \exp(2 \pi i/3)$ and the unitary operators $A_{x}^{k}$ and $B_y^l$ are defined as
\begin{equation}
A_{x}^k = \sum_{a = 0}^{2} \zeta^{-a k} \Pi_{x}^{a} \; \; \text{and} \; \; B_{y}^l = \sum_{b=0}^2 \zeta^{- b l} \Sigma_{y}^b.
\end{equation}
%for $\zeta = \exp{(2 \pi i/3)}$. 

There are six possible permutations that enter the game, namely the elements of group $\mathcal{S}_3 := \{ e, (01), (02), (12), (012), (021)\}$ where the permutations are denoted as usual in the cycle notation, i.e., $(01) = (0 \rightarrow 1, 1 \rightarrow 0, 2 \rightarrow 2)$, etc. We consider the maximal abelian subgroup $H$ of $\mathcal{S}_3$ ($H < \mathcal{S}_3$) given by $H := \{e, (012), (021) \}$. The corresponding left coset obtained by the action of $(01)$ on $H$ is given by $(01) H = \{(01), (12), (02)\}$. In this simple case of three outputs, we immediately see that the  permutation constraints in the game $b = \pi_{x,y}(a)$ for $\pi_{x,y} \in H$ are simply constraints of the form $a \ominus b \; \text{mod 3} = f^{\ominus}_{\pi_{x,y}}(x,y)$ with the correspondence $(f^{\ominus}_{e}(x,y) = 0, f^{\ominus}_{(012)}(x,y) = 2, f^{\ominus}_{(021)}(x,y) = 1)$. Similarly, permutation constraints $b = \pi_{x,y}(a)$ for $\pi_{x,y} \in (01) H$ correspond to constraints of the form $a \oplus b \; \text{mod 3} = f^{\oplus}_{\pi_{x,y}}(x,y)$. We have that 
\begin{eqnarray}
&&P(b = \pi_{x,y}(a) | x,y) = P(a \ominus b = f^{\ominus}_{\pi_{x,y}} | x, y) \nonumber \\
&& \; \; = \frac{1}{9} \sum_{k, l=0}^{2} \sum_{a=0}^{2} \zeta^{a (k+l) + f^{\ominus}_{\pi_{x,y}}(x,y) l} \langle A_x^k \otimes B_y^l \rangle  \nonumber \\
&& \; \; = \frac{1}{3} \sum_{k=0}^{2} \zeta^{- k f^{\ominus}_{\pi_{x,y}}} \langle A_x^k \otimes B_y^{-k} \rangle, \; \; \forall \pi_{x,y} \in H,
\end{eqnarray}  
and similarly, 
\begin{eqnarray}
&&P(b = \pi_{x,y}(a) | x,y) = P(a \oplus b = f^{\oplus}_{\pi_{x,y}} | x, y) \nonumber \\
&& \; \; = \frac{1}{9} \sum_{k, l=0}^{2} \sum_{a=0}^{2} \zeta^{a (k-l) + f^{\oplus}_{\pi_{x,y}}(x,y) l} \langle A_x^k \otimes B_y^l \rangle \nonumber \\
&&\; \; = \frac{1}{3} \sum_{k=0}^{2} \zeta^{k f^{\oplus}_{\pi_{x,y}}} \langle A_x^k \otimes B_y^k \rangle, \; \; \forall \pi_{x,y} \in (01) H. 
\end{eqnarray}  
Now, let us define vectors $|\alpha_{k} \rangle, |\beta_{k} \rangle$ as
\begin{eqnarray}
&&|\alpha_{k} \rangle := \sum_{x \in Q_A} \left((A_{x}^{k})^{\dagger} \otimes \identity \right) |\Psi \rangle \otimes |x \rangle, \; \; \nonumber \\
&&|\beta_{l} \rangle := \sum_{y \in Q_B} \left(\identity \otimes B_{y}^{k} \right) |\Psi \rangle \otimes |y \rangle,
\end{eqnarray}
and let us consider the "game matrices" defined in Eqs. (\ref{gamematrix1})  (\ref{gamematrix2}), where
$q^H(x,y) = q(x,y)$ for input pairs such that the winning constraint for this input pair comes from $H$, $\pi_{x,y} \in H$ and $q^H(x,y) = 0$ for the remaining input pairs. The distribution $q^{(01) H}(x,y)$ is defined analogously (as equal to $q(x,y)$ for pairs where the winning constraint comes from $(01) H$ and $0$ otherwise). Consequently, note that the matrices $\Phi^{H}_{k}$ and $\Phi^{(01) H}_{k}$ are such that $(\Phi^{H}_{k})_{i,j} \neq 0$ only if  $(\Phi^{(01) H}_{k})_{i,j} = 0$ and vice versa. 

Now, we observe that the quantum value of the unique game can be written as
\begin{widetext}
\begin{eqnarray}
\omega_q(\textsl{g}^{u}) &=& \frac{1}{3} \sum_{k=0}^{2} \left[ \sum_{(x,y) \in Q_A \times Q_B} q^H(x,y) \zeta^{-k f^{\ominus}_{\pi_{x,y}}} \langle A_x^k \otimes B_y^{-k} \rangle + q^{(01) H}(x,y) \zeta^{k f^{\oplus}_{\pi_{x,y}}(x,y)} \langle A_x^k \otimes B_y^k \rangle \right] \nonumber \\
& = & \frac{1}{3} \left[ 1 + \langle \alpha_1 | \mathbf{1} \otimes \Phi^{H}_1 | \beta_2 \rangle +  \langle \alpha_1 | \mathbf{1} \otimes \Phi^{(01) H}_1 | \beta_1 \rangle +\langle \alpha_2 | \mathbf{1} \otimes \Phi^{H}_2 | \beta_1 \rangle + \langle \alpha_2 | \mathbf{1} \otimes \Phi^{(01)H}_2 | \beta_2 \rangle \right] \nonumber \\
&\leq & \frac{1}{3} \left[1 +  \sqrt{m_Am_B}\| \Phi^{H}_1, \Phi^{(01)H}_1 \|_{\text{gen}} +  \sqrt{m_Am_B} \| \Phi^{H}_2, \Phi^{(01)H}_2 \|_{\text{gen}} \right].
\end{eqnarray} 
\end{widetext}
Here, $\| \Phi^{H}_k, \Phi^{(01)H}_k \|_{\text{gen}}$ is defined as in Eq. (\ref{eq:gen-norm}).
Compare with the usual spectral norm of a matrix $\| A \| = \max \{ \|A x \| : \| x \| = 1 \}$. 
\end{proof}

For more outcome unique games, the above bound can be generalized by identifying maximal abelian subgroups of the group of permutations appearing in the game, this approach will be pursued elsewhere. On the other hand,
when restricted to the scenario of linear games with $m_A, m_B$ inputs, the above method of proof recovers the bound that we derived in \cite{RAM}, namely
\begin{equation}
\label{bound}
\omega_q(\textsl{g}^{l}) \leq \frac{1}{d} \left[ 1 + \sqrt{m_A m_B} \sum_{k=1}^{d-1} \| \Phi_k \| \right]
\end{equation} 
with game matrices given by 
\begin{equation}\label{Ziobro}
\Phi_k = \sum_{x=1}^{m_A} \sum_{y=1}^{m_B} q(x,y) \zeta^{k f(x,y)} | x \rangle \!\langle y |.
\end{equation}

The advantage of formulating bounds on the quantum value in terms of game matrices is that it allows for an easy sufficient condition to recognize and construct games with no quantum advantage.  
\begin{prop}
\label{prop:qeqc}
Let $G$ be a linear game with $d$ outputs and let $\Phi_1, \dots, \Phi_{d-1}$ denotes its corresponding game matrices. If the maximum left and right singular vectors $|u_1 \rangle, |v_1 \rangle$ of $\Phi_1$ are composed of roots of unity entries, and if in addition the maximum singular vectors of $\Phi_k$ are obtained from $|u_1 \rangle, |v_1 \rangle$ by the substitution $\zeta \rightarrow \zeta^k$, then $\omega_q(G) = \omega_c(G)$.  
\end{prop}
\begin{proof}
As in the proof of the norm-based bounds on the quantum value, we use the generalized correlators $\langle A_x^k \otimes B_y^l \rangle$ to express the probabilities $P(a,b|x,y)$ as
\begin{equation}
P(a,b|x,y) = \frac{1}{d^2} \sum_{k,l=0}^{d-1} \zeta^{a k + b l} \langle A_x^k \otimes B_y^l \rangle,
\end{equation} 
so that the quantum success probability in the linear game can be expressed as
\begin{equation}
\omega_q^{(l)}(G) = \max_{\{\Pi_{x}^a \}, \{\Sigma_{y}^b \}}  \frac{1}{d} \sum_{k=0}^{d-1} \left[ \sum_{x=1}^{m_A} \sum_{y=1}^{m_B} q(x,y) \zeta^{k f(x,y)} \langle A_x^k \otimes B_y^k \rangle \right],
\end{equation}
with $A_x^k$ and $B_y^l$ given in general by 
\begin{equation}
A_x^k = \sum_{a=0}^{d-1} \zeta^{-a k} \Pi_{x}^a,\quad
B_y^l = \sum_{b=0}^{d-1} \zeta^{- b l} \Sigma_{y}^b.
\end{equation}
In terms of the matrices $\Phi_k$ defined in Eq. (\ref{Ziobro}) and the vectors $| \alpha_k \rangle = \sum_{x=1}^{m_A} ((A_x^k)^{\dagger} \otimes \textbf{1})| \psi \rangle \otimes |x \rangle$, $| \beta_l \rangle = \sum_{y=1}^{m_B} (\textbf{1} \otimes B_y^l)| \psi \rangle \otimes |y \rangle$, we can equivalently write
\begin{equation}
\omega_q^{(l)}(G) =  \max_{\{| \alpha_k \rangle \}, \{| \beta_k \rangle \}} \frac{1}{d} \left[1 + \sum_{k=1}^{d-1} \langle \alpha_k | \mathbf{1} \otimes \Phi_k | \beta_k \rangle \right].
\end{equation}

The classical success probability $\omega_c^{(l)}(G)$ is achieved by a deterministic strategy where Alice returns a set of deterministic outcomes $\{\hat{a}_x\}$ and similarly Bob returns $\{ \hat{b}_y \}$ upon receiving their respective inputs $x$ and $y$. Noting that $P_c(a \oplus_d b = f(x,y) | x,y) = \frac{1}{d} \sum_{k=0}^{d-1} \zeta^{k(f(x,y) - \hat{a}_x - \hat{b}_y)} = 1$ if $f(x,y) = \hat{a}_x \oplus_d \hat{b}_y$ and $0$ otherwise, 
we see that $\omega_c^{(l)}(G)$ can be written as
\begin{equation}
\omega_c^{(l)}(G) = \max_{\{\hat{a}_x\}, \{\hat{b}_y\}} \frac{1}{d} \sum_{k=0}^{d-1} \left[ \sum_{x=1}^{m_A} \sum_{y=1}^{m_B} q(x,y) \zeta^{k (f(x,y) - \hat{a}_x - \hat{b}_y)} \right].
\end{equation}
It is then readily seen that when the conditions stated in the proposition are met, the left and right singular vectors $|u_1 \rangle$ and $|v_1 \rangle$ of $\Phi_1$ corresponding to the maximum singular value define a consistent classical strategy that achieves the bound on the quantum value.  
\end{proof}
%When the game matrices $\Phi_k$ are Hermitian, we remark that one may check that the maximum eigenvectors of $\Phi_k$ satisfy the conditions in Proposition \ref{prop:qeqc}. 
The sufficient condition in Proposition \ref{prop:qeqc} is useful in constructing novel inequalities where quantum theory offers no advantage over classical theories. An example of a linear game with four outputs with no quantum advantage (that does not belong to the $NLC_d$ class explained below) is as follows. 
\[
\Phi_{ex} = (1/56) \begin{bmatrix}
    7       & -3 & 3 i & i \\
    -3      & 7 & i & 3 i \\
    -3 i & - i & 7 & -3  \\
    - i  & -3 i & -3 & 7
\end{bmatrix}\]
This game matrix $\Phi_{ex}$ has maximum eigenvector $[-i,i,-1,1]$, and a classical strategy where Alice and Bob ouput as their four outputs $[1,3,0,2]$ and $[3,1,0,2]$ respectively achieves the quantum success probability in the game.

\subsection{Non-local computation games}

The paradigmatic class of inequalities with no quantum advantage are given by the non-local computation games \cite{NLC}.
These concern the 
distributed computation of a boolean function 
%
%\begin{equation}
$f:\{0,1\}^n\mapsto\{0,1\}$
%\end{equation}
%
mapping $n$ bits to a single bit. Consider
that two parties, Alice and Bob, receive input strings of length $n$, 
$\mathbf{x}_n=(x_1,\ldots,x_n)$ and $\mathbf{y}_n=(y_1,\ldots,y_n)$, respectively. 
Each bit $x_i$ and $y_i$ is distributed with equal probability, ensuring that 
neither Alice nor Bob is able to correctly learn the bit $z_i=x_i\oplus_2 y_i$, where $\oplus_2$ stands for addition modulo two. 
Now, in order to 
perform distributed computation of the function $f$, Alice 
and Bob must output bits $a$ and $b$ such that 
\begin{equation}
\label{eq:binary-nlc-fn}
a\oplus b=f(x_1\oplus y_1,\ldots,x_n\oplus y_n)\equiv f(\mathbf{z}_n).
\end{equation}
To this end, they can agree on any strategy beforehand, however, after receiving
inputs, no communication can be exchanged between them.
Alice's and Bob's common aim is to maximize, for a 
given input probability distribution 
\begin{equation}
\label{eq:binary-nlc-pr}
p(\mathbf{x}_n,\mathbf{y}_n)=(1/2^n)p(\mathbf{x}_n\oplus_2 \mathbf{y}_n)\equiv(1/2^n)p(\mathbf{z}_n)
\end{equation}
the average probability of success in this task:
\begin{equation}
\omega=\frac{1}{2^n}\sum_{\textbf{x}_n,\textbf{y}_n\in\{0,1\}^n}p(\mathbf{z}_n)p(a\oplus b=f(\mathbf{z}_n)|\mathbf{x}_n,\mathbf{y}_n).
\end{equation}
In what follows we will denote this task by $NLC_2$.

One immediately realizes that the above game, fits very well with
the Bell scenario outlined in Sec. \ref{Preliminaries}. Now, $x,y=1,\ldots,2^n$ (encoded into 
$n$-bit strings $\mathbf{x}_n$ and $\mathbf{y}_n$) are the choices of measurements
that Alice and Bob can perform on some quantum state $\rho$
(which due to the convexity argument can be taken pure), while 
$a$ and $b$ stand for the measurements outcomes. 
It is then natural to ask 
how local $\mathcal{L}$ and quantum correlations $\mathcal{Q}$
perform at $NLC_2$, and, in particular, whether quantum theory provides any advantage in this task. Quite surprisingly, it turns out 
that this is not the case and for no quantum state and measurements 
the success probability $\omega$ can surpass the maximal success probability
over all classical correlations $\omega_{\mathrm{c}}=\max_{\mathbf{p}\in P_{\mathrm{L}}}\omega$ \cite{NLC}. In other words, 
for any such task, $\omega_{\mathrm{q}}=\omega_{\mathrm{c}}$ with
$\omega_{\mathrm{q}}=\max_{\mathbf{p}\in \mathcal{Q}}\omega$.

On the other hand, there exist super-quantum correlations 
obeying the no-signaling principle that can still surpass $\omega_{\mathrm{q}}$. 
So, there might exist a more general theory, respecting the 
no-signaling principle, which at this task can be more powerful than
quantum theory. This observation, being so striking, was then exploited 
to propose a principle to pick out the quantum $\mathcal{Q}$
from $\mathcal{N}$: \textit{quantum correlations are those that perform no better than
classical ones at the $NLC_2$ task}.

In Ref. \cite{NLC} only games with binary outcomes 
were studied and they posed as an open question the consideration of functions $f$ with multi-bit outputs as well as functions with 
higher input and output alphabets. Such generalization was recently proposed 
in Ref. \cite{RAM}. Let us briefly recall it here. To this end, let us 
consider a function $f:\mathbb{Z}_d^n\mapsto\mathbb{Z}_d$ mapping $n$
dits to a single one with $\mathbb{Z}_d=\{0,\ldots,d-1\}$. Imagine then that 
Alice and Bob are given string of $n$ \textit{dits}, $\mathbf{x}_n$ and $\mathbf{y}_n$, 
and their aim is to provide dits $a,b\in\mathbb{Z}_d$ such that $a\oplus_d b=f(\mathbf{z}_n)$, 
where $z_i=x_i\oplus_d y_i$ and $\oplus_d$ denotes addition modulo $d$. We have demonstrated
in Ref. \cite{RAM} that for the class of functions $f$ given by 
\begin{eqnarray}\label{condition}
f(\mathbf{z}_n)&=&g(x_1\oplus_d y_1
,\ldots,x_{n-1}\oplus_d y_{n-1})\cdot (x_n\oplus_d y_n)\nonumber\\
&\equiv & g(\mathbf{x}_{n-1}\oplus_d \mathbf{y}_{n-1})\cdot (x_n\oplus_d y_n)
\end{eqnarray}
where $g:\mathbb{Z}_d^{n-1}\mapsto \mathbb{Z}_d$ is any function, 
and for probability distributions 
\begin{equation}\label{condition2}
p(\mathbf{x}_n,\mathbf{y}_n)
=\frac{1}{d^{n+1}}p(\mathbf{x}_{n-1}\oplus_d\mathbf{y}_{n-1}),
\end{equation}
%%\frac{1}{d^n}p(\mathbf{x}_n\oplus_d\mathbf{y}_n)
quantum correlations perform no better than classical one in 
maximizing the average success probability of winning this game
\begin{equation}
\omega=\frac{1}{d^{n}}\sum_{\mathbf{x}_n\mathbf{y}_n\in\mathbb{Z}_d^n}p(\mathbf{z}_n)p(a \oplus_d b=f(\mathbf{z}_n)|\mathbf{x}_n,\mathbf{y}_n).
\end{equation}
Below we denote games defined by Eqs. (\ref{condition})
and (\ref{condition2}) by $NLC_d$.

Interestingly, as we will see below, the restriction to functions given 
by Eq. (\ref{condition}) is vital for the proof that quantum correlations 
do not provide any advantage over the classical ones in $NLC_d$. In fact, otherwise
it is not difficult to find an example of a non-local computation task for 
which $\omega_{\mathrm{q}}> \omega_{\mathrm{cl}}$ if only $d\geq 3$. To this end, 
let us consider the case of $d=3$ and $n=1$ (each Alice and Bob receive a single trit
$x$ and $y$, respectively) and a function $f:\mathbb{Z}_3\mapsto \mathbb{Z}$
given by
\begin{equation}
f(z)=\left\{
\begin{array}{ll}
1,& z=2\\
0,& z=0,1.
\end{array}
\right.
\end{equation}
Assuming then that $p(x,y)=(1/3)p(x\oplus_3 y)=1/9$, it is not difficult to 
realize that the maximal classical value of this game amounts to $\omega_{\mathrm{cl}}=2/3$ 
(just by algorithmically checking all possible deterministic strategies). 
On the other hand, the game matrices 
corresponding to the above game are
\begin{equation}\label{exampleNLC3}
\Phi_1=\frac{1}{9}\left(
\begin{array}{ccc}
1 & 1 & \omega\\
1 & \omega &1 \\ 
\omega &1 &1
\end{array}
\right)
\end{equation}
and $\Phi_2=\Phi_1^*$, where the asterisk stands for the standard complex conjugation.
The operator norms of these matrices are $\|\Phi_1\|=\|\Phi_2\|=\sqrt{3}$, and consequently, our bound 
(\ref{bound}) implies that the maximal quantum value is upper-bounded as 
$\omega_{\mathrm{q}}\leq (1/3)[1+2\sqrt{3}/3]\approx 0.7182$.
Finally, taking the two-qutrit maximally entangled state 
$\ket{\psi_3^+}=(1/\sqrt{3})\sum_{i=0}^2\ket{ii}$ and optimizing 
over one-qutrit measurements on Alice and Bob sites, one finds 
that the maximal value is lower bounded as $\omega_{\mathrm{q}}\geq 0.7124$, 
thus, clearly, $\omega_{\mathrm{q}}>\omega_{\mathrm{cl}}$ for this game. 
In fact, the numerical investigations of \cite{Liang} show that the lower bound in fact gives the quantum value for this game 
\cite{footnote2}. 
This example in the three-outcome scenario at first glance appears rather surprising in view
of the result of Ref. \cite{NLC} stating that all binary non-local computation tasks
exhibit no quantum advantage, however there is a clear reason for this. For $d \geq 3$, one has functions of the form $(x_n \oplus_d y_n)^2$ which are equivalent to the $\text{CHSH}_d$ game $f_{\text{CHSH}_d}(x_n, y_n) = x_n \cdot y_n$ under local relabellings. Quantum theory gives an advantage in the $\text{CHSH}_d$ game so that we see that the restriction of the functions to the condition in Eq.(\ref{condition}) is necessary.

\section{Tightness of $NLC_d$}
In this section, we investigate the tightness of the paradigmatic class of two-party Bell inequalities with no quantum advantage, namely the $NLC_d$ inequalities.  
{In \cite{GYNI} the authors have shown that $NLC_2$ are not facet defining inequalities for the case of two and three input bits. However the proof of non-tightness for any number of input bits was left as an open problem. }

Recall that non-local computation is a game 
in which Alice and Bob receive $n$-dit strings $\mathbf{x}_n$ and $\mathbf{y}_n$
with $x_i,y_i\in\mathbb{Z}_d$ and must output dits $a$ and $b$ such that 
$a\oplus_d b=f(\mathbf{x}_{n}\oplus_d \mathbf{y}_{n})$
for some function $f$ fulfilling the condition (\ref{condition}). We additionally 
assume that the probability distribution $p(\mathbf{x}_n,\mathbf{y}_n)$ is given by Eq. (\ref{condition2}).
These games give rise to Bell inequalities which can explicitly be 
stated as
\begin{equation}\label{inequalities}
\frac{1}{d^n}\sum_{\mathbf{x}_n,\mathbf{y}_n\in\mathbb{Z}_d^n}p(\mathbf{z}_n)
p(a\oplus_d b=f(\mathbf{z}_n)|\mathbf{x}_n,\mathbf{y}_n)\leq \omega_{\mathrm{c}}.
\end{equation}
As already explained, one approach to show that an inequality
does not define a facet of the local polytope is to prove that there are fewer than $D$ affinely 
independent classical deterministic boxes saturating it, i.e. achieving the classical value. Here, we use an alternative approach based on the following straightforward observation, 
namely that the inequality is non-tight if we can exhibit a decomposition of the inequality into other inequalities that define supporting hyperplanes for the polytope. 

\begin{lemma}
\label{facet-lem}
If $P$ is a polytope, then the intersection of two faces of $P$ is a face of $P$. 
A facet of $P$ cannot be obtained as the intersection of two or more different faces of $P$. 
\end{lemma}

\begin{proof}
%\textcolor{red}{
Suppose $F$ and $G$ are two faces of $P$, so there are corresponding supporting hyperplanes $H_{(\textbf{c}_{F}, b_{F})}$ and $H_{(\textbf{c}_{G}, b_{G})}$ given as
\begin{eqnarray}
H_{(\textbf{c}_{F}, b_{F})} &:=& \{\textbf{p} | \textbf{c}_{F} \cdot \textbf{p} = b_{F}\} \nonumber \\
H_{(\textbf{c}_{G}, b_{G})} &:=& \{\textbf{p} | \textbf{c}_{G} \cdot \textbf{p} = b_{G} \},
\end{eqnarray}
such that $F = P \cap H_{(\textbf{c}_{F}, b_{F})}$ and $G = P \cap H_{(\textbf{c}_{G}, b_{G})}$. The halfspace 
\begin{equation}
\{ \textbf{p} | (\textbf{c}_{F} + \textbf{c}_{G}) \cdot \textbf{p} \leq b_{F} + b_{G}\}
\end{equation}
contains $P$ and for any $\textbf{p} \in P$, we have that $(\textbf{c}_{F} + \textbf{c}_{G}) \cdot \textbf{p} = b_{F} + b_{G}$ only when both $\textbf{c}_{F} \cdot \textbf{p} = b_{F}$ and $\textbf{c}_{G} \cdot \textbf{p} = b_{G}$.  
Hence the intersection of $F$ and $G$ is the intersection of $P$ with the hyperplane $H_{(\textbf{c}_{F} + \textbf{c}_{G}, b_{F} + b_{G})}$, and so $F \cap G$ is a face of $P$. If $F \cap G$ is a facet of $P$, then by the above argument, we have that $F$ and $G$ must also be facets of $P$. The affinely independent boxes $\textbf{p}$ that define $F \cap G$ also define $F$ and $G$ and so $F = G$. This shows that a facet cannot be obtained as the intersection of two (or more) differing faces. 
%}
\end{proof}

\subsection{$NLC_2$ XOR games} 
Our first aim now is to prove that none of the $NLC$ inequalities is tight in the case $d=2$, that is, 
such Bell inequalities do not give rise to facets of the corresponding local 
polytope. The NLC inequalities are also the interesting case of the XOR games that would serve as paradigmatic candidates 
for facets shared by the elliptope and the the cut polytope for the bipartite graph in the question posed by Avis et al. \cite{AII06} (cf. Section \ref{sec:cut-poly}). We will later see how to generalize this to arbitrary prime $d$ but for a restricted class of functions. 

%Recall that the $NLC$ games for $d=2$ 
Note that the $NLC$ games with $n=1$ (i.e., functions $k \cdot (x_n \oplus y_n) + \delta$ with $k,\delta \in \{0, 1\}$) 
are uninteresting because there is a simple classical strategy that wins these games, namely $a(x_n) = k \cdot x_n $ and $b(y_n) = k \cdot y_n+ \delta$ so that $w_c = w_q = w_{ns}$ for them. The interesting games are with $n > 1$ where $w_c = w_q < w_{ns} = 1$. The fact that $w_{ns} = 1$ simply follows from the fact that these games belong to the class of XOR games, for which a no-signaling strategy always exists to win the game, namely one where $P(a,b | \textbf{x}_n, \textbf{y}_n) = \frac{1}{2}$ when $a \oplus_2 b = f(\textbf{x}_n, \textbf{y}_n)$ and $0$ otherwise.

\begin{thm}
The binary outcome non-local computation game $NLC$ inequalities for arbitrary functions $f(x_1 \oplus y_1, \dots, x_n \oplus y_n)$ (Eq.(\ref{eq:binary-nlc-fn})) and arbitrary probability distributions $p(\textbf{x}_n, \textbf{y}_n) = \frac{1}{2^n}p(x_1 \oplus y_1, \dots, x_n \oplus y_n)$ (Eq.(\ref{eq:binary-nlc-pr})) do not define facets of the local Bell polytope or the correlation polytope for any number $n > 1$ of input bits. 
\end{thm}
\begin{proof}
The idea of the proof is to exhibit a decomposition of the NLC game Bell inequality for function $f(x_1 \oplus y_1, \dots, x_n \oplus y_n)$ and probability distribution $\frac{1}{2^n} p(x_1 \oplus y_1, \dots, x_n \oplus y_n)$ as a sum of inequalities that are valid for the Bell polytope and the correlation polytope so that by Lemma \ref{facet-lem}, the NLC inequality cannot define a facet of either of these polytopes. More precisely, we identify two subgames $NLC^{x_1 = 0}_{\text{sub}}$ and $NLC^{x_1 = 1}_{\text{sub}}$ that define inequalities $\textbf{c}_{NLC_0} \cdot \textbf{p} \leq \omega_{c}(NLC^{x_1 = 0}_{\text{sub}})$ and $\textbf{c}_{NLC_1} \cdot \textbf{p} \leq \omega_{c}(NLC^{x_1 = 1}_{\text{sub}})$ that themselves can be violated in quantum theory, but which sum up to the NLC Bell inequality which admits no quantum violation, i.e., 
\begin{eqnarray}
\label{eq:nlc-decomp}
\textbf{c}_{NLC} &=& \textbf{c}_{NLC_0} + \textbf{c}_{NLC_1} \nonumber \\
\omega_{c}({NLC}) &=& \omega_{c}(NLC^{x_1 = 0}_{\text{sub}}) + \omega_{c}(NLC^{x_1 = 1}_{\text{sub}}). 
\end{eqnarray}

We define $NLC^{x_1 = j}_{\text{sub}}$ as the subgame of the $NLC$ game with $x_1$ fixed to the value $j$ for $j=0,1$. In other words, $NLC^{x_1 = j}_{\text{sub}}$ corresponds to the function $f(j \oplus y_1, x_2 \oplus y_2, \dots, x_n \oplus y_n)$ and probability distribution $p(j, x_2, \dots, x_n, y_1, \dots, y_n) = \frac{1}{2^n}p(j \oplus y_1, x_2 \oplus y_2, \dots, x_n \oplus y_n)$. Notice that the maximum algebraic and no-signaling value of the $NLC^{j}_{\text{sub}}$ Bell inequality is $\frac{1}{2}$. Clearly, by definition $\textbf{c}_{NLC} = \sum_{j=0,1} \textbf{c}_{NLC_j}$, our aim is now to show that $\omega_{c}(NLC^{x_1 = j}_{\text{sub}}) = \frac{1}{2} \omega_{c}({NLC})$ so that Eq.(\ref{eq:nlc-decomp}) holds. We do this by showing that the optimal classical strategy for the subgame $NLC^{x_1 = j}_{\text{sub}}$ is the same as the optimal classical strategy for the game $NLC$ restricted to the case $x_1 = j$.  

\begin{lemma}
\label{lem:class-strat-nlc}
$\omega_{c}(NLC^{x_1 = j}_{\text{sub}}) = \frac{1}{2} \omega_{c}({NLC})$.
\end{lemma}
\begin{proof}
The game matrix $\Phi^{x_1 =j}$ for the subgame $NLC^{x_1 = j}_{\text{sub}}$ is given by the entries
\begin{eqnarray}
\Phi^{x_1 = j}_{(j,x_2 \dots, x_n), (y_1, \dots, y_n)} &=& \frac{1}{2^n} (-1)^{f(j \oplus y_1, x_2 \oplus y_2, \dots, x_n \oplus y_n)} \nonumber \\ &&\times p(j \oplus y_1, x_2 \oplus y_2, \dots, x_n \oplus y_n). \nonumber \\
\end{eqnarray} 
We now further consider the subgames $NLC^{x_1 = j, y_1 = k}_{\text{sub}}$ with $k = 0,1$, i.e., with game matrices $\Phi^{x_1 = j, y_1 =k}$ that have entries
\begin{eqnarray}
\Phi^{x_1 = j, y_1 =k}_{(j,x_2 \dots, x_n), (k,y_2 \dots, y_n)} &=& \frac{1}{2^n} (-1)^{f(j \oplus k, x_2 \oplus y_2, \dots, x_n \oplus y_n)} \cdot \nonumber \\ && p(j \oplus k, x_2 \oplus y_2, \dots, x_n \oplus y_n). \nonumber \\
\end{eqnarray} 
Now, notice that the subgame $NLC^{x_1 = j, y_1 = k}_{\text{sub}}$ itself has the structure of an NLC game with $n-1$ bits input to each party. This in particular implies by the results of \cite{NLC} that the corresponding game matrix $\Phi^{x_1 = j, y_1 =k}$ is diagonal in the basis formed by the (normalized) Hadamard vectors $|u_l\rangle$ with $l = 1, \dots, 2^{n-1}$,
\begin{eqnarray}
\Phi^{x_1 =j, y_1 = k} = \sum_{l=1}^{2^{n-1}} \lambda^{x_1 =j, y_1 =k}_l | u_l \rangle \langle u_l |.
\end{eqnarray}
The maximum no-signaling value of the subgame $NLC^{x_1 = j, y_1 = k}_{\text{sub}}$ is 
\begin{equation}
\omega_{ns}(NLC^{x_1 = j, y_1 = k}) = \sum_{x_2, \dots, x_n, y_2, \dots, y_n} \Phi^{x_1 = j, y_1 =k}_{(j,x_2 \dots, x_n), (k,y_2 \dots, y_n)}.
\end{equation} 
As shown in \cite{NLC} the optimal classical strategy for the game $NLC^{x_1 = j, y_1 = k}_{\text{sub}}$ is given  by the eigenvector $|u_{l_{max}}\rangle$ corresponding to the eigenvalue of maximum absolute value $\vert \lambda^{x_1 = j, y_1 = k}_{l_{max}} \vert$. Alice outputs according to $(-1)^{a(x_1=j, x_2, \dots,x_n)} = |u_{l_{max}}\rangle_{x_2, \dots, x_n}$ and Bob outputs according to $(-1)^{b(y_1=k, y_2, \dots,y_n)} = sgn(\lambda^{x_1 =j, y_1 =k}_{l_{max}}) \cdot |u_{l_{max}}\rangle_{y_2, \dots, y_n}$. We therefore see that the game matrix we are interested in $\Phi^{x_1 = j}$ can be written as
%\begin{equation}
$\Phi^{x_1 = j} = \sum_{l=1}^{2^{n-1}} |u_l \rangle \langle v_l|$,
%\end{equation}
where 
%\begin{equation}
$|v_l \rangle = |w_l\rangle \otimes |u_l\rangle$, 
%\end{equation}  
with 
\begin{equation}
|w_l \rangle = \begin{bmatrix} \lambda^{x_1 = j, y_1 =0}_l \\  \lambda^{x_1 = j, y_1 =1}_l \end{bmatrix}.
\end{equation}
%
%(\lambda^{x_1 = j, y_1 =0}_l, \lambda^{x_1 = j, y_1 =1}_l)^T$. 
The optimal classical value of the game $NLC^{x_1 = j}_{\text{sub}}$ is given as 
\begin{eqnarray}
\omega_c(NLC^{x_1 = j}_{\text{sub}}) &=& \frac{1}{2}\left[ 1 + \max_{|s_A \rangle, |s_B \rangle} \langle s_A | \Phi^{x_1 =j} | s_B \rangle \right]\nonumber \\
&=& \frac{1}{2}\left[1  + \max_{|s_A \rangle, |s_B \rangle} \sum_{l=1}^{2^{n-1}} \langle s_A | u_l \rangle \langle v_l | s_B \rangle\right], \nonumber \\
\end{eqnarray}
where the maximization is over all vectors $|s_A \rangle, |s_B \rangle$ with $\pm 1$ entries. 

Now, any general strategy vector can be written in terms of the basis formed by the Hadamard vectors, therefore 
the classical strategies $| s_A \rangle, |s_B \rangle$ can be writen as %in terms of the basis formed by the Hadamard vectors, i.e.,
\begin{eqnarray}
|s_A \rangle &=& \sqrt{2^{n-1}}  \sum_{i_A=1}^{2^{n-1}} \gamma^{A}_{i_A} | u_{i_A} \rangle, \nonumber \\
|s_B \rangle &=&  \sqrt{2^{n-1}} \sum_{i_B=1}^{2^{n-1}} \sum_{\mu_B=0,1} \gamma^{B}_{i_B,\mu_B} \begin{bmatrix}
    1 \\ 
    (-1)^{\mu_B}
\end{bmatrix}
%(1, (-1)^{\mu_B} 1)^T 
\otimes |u_{i_B} \rangle,
\end{eqnarray}
where, $\sum_{i_A} |\gamma^{A}_{i_A}|^2=1$ and $\sum_{i_B, \mu_B} |\gamma^{B}_{i_B,\mu_B}|^2=1$, 
reflects the fact that deterministic classical strategy vectors for Alice and Bob have norms $2^{n-1}$ and $2^n$ respectively.
%Note that any general strategy vector can be written in the basis formed by the Hadamard vectors.
Optimizing over such strategies gives us
\begin{widetext}
\begin{eqnarray}
\label{eq:cl-strat}
\omega_c(NLC^{x_1 = j}_{\text{sub}})
&\leq& \frac{1}{2} \left[1+2^{n-1} \max_{\{\gamma^{A}_{i_A}\}, \{\gamma^{B}_{i_B, \mu_B}\}} \sum_{i_A, i_B, \mu_B} (\gamma^{A}_{i_A})^{*} \gamma^{B}_{i_B, \mu_B} \delta_{i_A, i_B} \cdot  \left|  \lambda^{x_1 = j, y_1 =0}_{i_B} \pm  \lambda^{x_1 = j, y_1 =1}_{i_B} \right| \right], \nonumber \\
&\leq &\frac{1}{2} \left[1+2^{n-1} \max_l \sum_{k=0,1} \vert \lambda^{x_1 =j, y_1 =k}_{l} \vert \right].
\end{eqnarray}
\end{widetext}
where we have used the constraints $\sum_{i_A} |\gamma^{A}_{i_A}|^2=1$ and $\sum_{i_B, \mu_B} |\gamma^{B}_{i_B,\mu_B}|^2=1$ and the Cauchy-Schwarz inequality to bound the value.  

The upper bound in Eq.(\ref{eq:cl-strat}) is achieved by choosing a Hadamard strategy, i.e., with 
\begin{eqnarray}
\label{eq:cl-strat-2}
| s_A \rangle &=& \sqrt{2^{n-1}} | u_l \rangle, \nonumber \\
|s_B \rangle &=& \sqrt{2^{n-1}} \begin{bmatrix} sgn(\lambda^{x_1 =j, y_1 =0}_{l}) \\ sgn(\lambda^{x_1 =j, y_1 =1}_{l})  \end{bmatrix} \otimes |u_l \rangle
\end{eqnarray}
 for the value of $l$ that maximizes $\sum_{k=0,1} \vert \lambda^{x_1 =j, y_1 =k}_{l} \vert $. The maximum value thus achieved is 
\begin{equation}
\omega_{c}(NLC^{x_1 = j}_{\text{sub}}) = \frac{1}{2} \left[1 + 2^{n-1} \sum_{k=0,1} \vert \lambda^{x_1 =j, y_1 =k}_{l} \vert \right].
\end{equation} 
Now, the NLC game matrices have the symmetry \cite{NLC} that $\Phi^{x_1 = j, y_1 = k} = \Phi^{x_1 = j \oplus 1, y_1 = k \oplus 1}$. 
This implies that the game matrices $\Phi^{x_1 =j, y_1 = k}$ for $j=0,1$ are equivalent to each other under a relabeling of the inputs for Bob, so that we have 
\begin{equation}
\label{eq:subgame-eq}
\omega_{c}(NLC^{x_1 = 0}_{\text{sub}}) = \omega_{c}(NLC^{x_1 = 1}_{\text{sub}}).
\end{equation} 

Moreover, we have shown that the maximum classical value of both the subgames $NLC^{x_1 = j}_{\text{sub}}$ for $j=0,1$ is
achieved by a Hadamard strategy in Eq.(\ref{eq:cl-strat-2}), i.e., $|s_A^{j} \rangle$ is given by a Hadamard vector of 
length $2^{n-1}$ and $|s_B^{j} \rangle$ is given by a Hadamard vector of length $2^n$. As mentioned,
the $NLC$ game matrices have the property that $\Phi^{x_1 = j, y_1 = k} = \Phi^{x_1 = j \oplus 1, y_1 = k \oplus 1}$ so 
that $sgn(\lambda^{x_1 =j, y_1 = k}_{l}) = sgn(\lambda^{x_1 =j \oplus 1, y_1 = k \oplus 1}_{l})$ giving that Bob's optimal 
strategy vector $|s_B^{j} \rangle$ remains the same for $j=0,1$ up to an overall $\pm$ sign. This implies that Alice's optimal
strategy vector $|s_A^{j} \rangle$ to achieve $\omega_c(NLC^{x_1 = j}_{\text{sub}})$ also remains the same for $j=0,1$ up to 
a $\pm$ sign. Since the direct sum of the two Hadamard vectors of length $2^{n-1}$ is a Hadamard vector of length
$2^n$, this defines a Hadamard strategy for both Alice and Bob for the game $NLC$ that achieves the value 
\begin{equation}
\label{eq:nlc-opt-strat}
\omega_c(NLC) \geq \omega_{c}(NLC^{x_1 = 0}_{\text{sub}}) + \omega_{c}(NLC^{x_1 = 1}_{\text{sub}}).
\end{equation}
Since by definition of a subgame, $\omega_c(NLC) \leq \omega_{c}(NLC^{x_1 = 0}_{\text{sub}}) + \omega_{c}(NLC^{x_1 = 1}_{\text{sub}})$, we have equality in Eq.(\ref{eq:nlc-opt-strat}). Then using Eq.(\ref{eq:subgame-eq}), we obtain the statement of the Lemma. 
%Finally, remark that the quantum value of the games $NLC^{x_1=j}_{\text{sub}}$ themselves does not have to equal to its classical value 
\end{proof}
Remark that the quantum value of the games $NLC^{x_1=j}_{\text{sub}}$ themselves does not have to equal to their classical 
value, as the bound in Eq.(\ref{eq:cl-strat}) of Lemma \ref{lem:class-strat-nlc} is derived through a maximization over 
$\pm 1$ vectors $|s_A \rangle, |s_B\rangle$ and hence only applies to the classical value and not the quantum value 
of these games. By Lemma \ref{lem:class-strat-nlc} then we 
have $\omega_{c}(NLC^{x_1 = j}_{\text{sub}}) = \frac{1}{2} \omega_{c}({NLC})$, so that Eq.(\ref{eq:nlc-decomp}) holds, giving that the $NLC$ games do not constitute facets of the Bell polytope. Since the decomposition in Eq.(\ref{eq:nlc-decomp}) is into XOR game inequalities that are valid for the correlation polytope as well (since they do not involve any local marginal terms), the $NLC$ game Bell inequalities do not constitute facets of the correlation polytope either. 
\end{proof}

\subsection{$NLC_d$ LINEAR games}
Let us now proceed to investigate the tightness of $NLC_d$ games defined by the functions $f$ in Eq.(\ref{condition}) and probability distributions $p(\mathbf{x}_n,\mathbf{y}_n)$ given by Eq. (\ref{condition2}) for arbitrary prime $d$ and arbitrary $n$ input dits. As we have already seen, this restriction on the functions is necessary for these inequalities to define faces of the quantum set $\mathcal{Q}$. 

As shown in \cite{RAM}, these games are composed of $d$ different subgames that are the building blocks of the $NLC_d$ game matrices $\Phi_k$ with $k \in \{1, \dots, d-1\}$. These subgames take the form $G(t) := \{a \oplus_d b = t \cdot (x \oplus_d y)\}$ where $t \in \{0, \dots, d-1\}$, and they can be readily seen to be won by a classical strategy where Alice outputs $a(x) = t \cdot x$ and Bob outputs $b(y) = t \cdot y$. For $i \in \{0, \dots, d-1\}$, let $\lambda_{NLC_d}(i)$ denote the weighted number of times the game $G(i)$ occurs in the first block of $d$ rows of the $NLC_d$ game matrix $\Phi_1$, i.e.
\begin{equation}
\label{eq:lambda-d}
\lambda_{NLC_d}(i) :=\sum_{\stackrel{\textbf{y}_{n-1} }{\text{ s.t.} g(y_1, \dots, y_{n-1}) = i}} p(y_1, \dots, y_{n-1}).
\end{equation}
Let $\Lambda_{NLC_d} := \max_{i \in \{0, \dots, d-1\}} \lambda_{NLC_d}(i)$ and let $i_{\text{max}}$ denote the value of $i$ for which this maximum is achieved. Observe that $1/d \leq \Lambda_{NLC_d} \leq 1$. We showed in \cite{RAM} that the optimal classical strategy in the $NLC_d$ game is for Alice to output $a(\mathbf{x}_n) = i_{\text{max}} \cdot x_n$ and for Bob to output $b(\mathbf{y}_n) = i_{\text{max}} \cdot y_n$. Such a strategy achieves the optimal value for the game
\begin{equation}
\omega_c(NLC_d) = \omega_q(NLC_d) = \frac{1}{d} \left( 1 + (d-1) \Lambda_{NLC_d} \right).
\end{equation} 
We now show that for games with $\Lambda_{NLC_d} \geq 1/2$, a decomposition into other valid inequalities for the corresponding Bell polytope can be found so that the $NLC_d$ games with this property also do not define facets of the Bell polytope.

\begin{prop}
Any $NLC_d$ game for arbitrary prime $d$ and arbitrary number of input dits $n$ satisfying $\Lambda_{NLC_d} \geq 1/2$ does not define a facet of the corresponding Bell polytope.   
\end{prop}
\begin{proof}
As before, the proof works by a decomposition into valid face-defining inequalities $NLC_d^{s}$ of the corresponding Bell polytope. The subgames $NLC_d^{s}$ correspond to a function $g(\tilde{x}_1 \oplus_d y_1, \dots, \tilde{x}_{n-1} \oplus_d y_{n-1}) \cdot (x_{n} \oplus_d y_{n})$ with fixed inputs $\tilde{x}_i \in \{0,\dots, d-1\}$, which for definiteness we fix to $\tilde{x}_i = 0$ for $i \in \{0, \dots, n-1\}$. In other words, in the sub-game $NLC_d^{s}$, Alice receives a single dit input $x_n$ while Bob receives $n$ dits $y_1, \dots, y_n$ with probability $\frac{1}{d^{n+1}} p(y_1, \dots, y_{n-1})$. The probabilities in the sub-game do not sum to $1$, the maximum no-signaling value of the sub-game is $\frac{1}{d^{n-1}}$ with a factor of $d^2$ coming from the $d^2$ choices of $(x_n, y_n)$ for fixed $(y_1, \dots, y_{n-1})$ which form the $d \times d$ subgames $G(j)$.

Recall that $\Lambda_{NLC_d} := \max_{i \in \{0, \dots, d-1\}} \lambda_{NLC_d}(i)$ where $\lambda_{NLC_d}(i)$ is given by Eq.(\ref{eq:lambda-d}). We show now that when $\Lambda_{NLC_d} \geq \frac{1}{2}$, an optimal classical strategy for each of the subgames $NLC_d^{s}$ is identical to the optimal classical strategy for the game $NLC_d$ itself, giving that the Bell inequality corresponding to the game $NLC_d$ can be decomposed into valid face-defining inequalities corresponding to the subgames $NLC_d^{s}$ as in Lemma \ref{facet-lem}. As explained earlier, from \cite{RAM} we know that the optimal classical strategy in the $NLC_d$ game is for Alice to output $a(\mathbf{x}_n) = i_{\text{max}} \cdot x_n$ and for Bob to output $b(\mathbf{y}_n) = i_{\text{max}} \cdot y_n$. This strategy for the $NLC_d^{s}$ subgame gives the following lower bound
\begin{equation}
\label{eq:class-nlcd-sub}
\omega_c(NLC_d^{s}) \geq \frac{1}{d^n} \left( 1 +  (d-1) \Lambda_{NLC_d} \right). 
\end{equation}

Let us now show that any other classical strategy for $NLC_d^{s}$ does not achieve a larger value than in (\ref{eq:class-nlcd-sub}) so that the inequality there is in fact an equality. To do this, we examine classical strategies for the subgames $G(t)$ defined by the winning constraint $\{a \oplus_d b = t \cdot (x \oplus_d y)\}$. Fix any $t \in \{0, \dots, d-1\}$ and consider any arbitrary deterministic strategy for Alice, i.e., a set of Alice's deterministic outputs $a(x) \in \{0, \dots, d-1\}$. For $k \in \{0, \dots, d-1\}$, let $\mathcal{A}^{(t)}_k$ denote the set of inputs for which Alice outputs according to an optimal strategy, i.e., $\mathcal{A}^{(t)}_k = \{x : a(x) = t \cdot x \oplus_d k\}$, let $\mathcal{A}^{(t)}_{k_{\text{max}}}$ denote the set of maximum cardinality $m_t = |\mathcal{A}^{(t)}_{k_{\text{max}}}| = \max_{k} | \mathcal{A}^{(t)}_k|$. Note that $1 \leq m_t \leq d$.  
%denote the cardinality of this set.  
%For $k \in \{0, \dots, d-1\}$, let $\mathcal{A}_{k} = \{ x : a(x) = k\}$ and let $\mathcal{A}_{k_{\text{max}}}$ denote the set $\mathcal{A}_{k}$ of maximum cardinality $m = \max_{k} |\mathcal{A}_k|$. 
Bob's optimal strategy for the game $G(t)$ is then clearly seen to be to output $b(y) = t \cdot y \ominus_d k$, such a strategy satisfies $m_t d$ of the $d^2$ winning constraints in $G(t)$. Now, given Alice's deterministic strategy, consider game $G(t')$ for any other value of $t' \neq t$ and the analogous sets $\mathcal{A}^{(t')}_{k'}$ with $m_{t'} = |\mathcal{A}^{(t')}_{{k'}_{\text{max}}}| = \max_{k'} | \mathcal{A}^{(t')}_{k'}|$. Due to the fact that for any $k' \in \{0, \dots, d-1\}$, the equation
\begin{equation}
t \cdot x \oplus_d k_{\text{max}} = t' \cdot x \oplus_d k',
\end{equation} 
is satisfied for at most one value of $x \in \{0, \dots, d-1\}$, we have that $m_{t'} \leq d - m_{t} + 1$. This gives that for any deterministic strategy by Bob for the game $G(t')$, at most $(d - m_t + 1) d$ of the $d^2$ winning constraints in $G(t')$ can be satisfied. 

Now the subgame $NLC_d^{s}$ is composed of $\Lambda_{NLC_d}$ blocks of $G(t)$ for some fixed $t$ and $(1- \Lambda_{NLC_d})$ blocks of games $G(t')$ for $t' \neq t$. The above analysis then implies that any classical deterministic strategy for $NLC_d^{s}$ can at best achieve the value
\begin{equation}
\label{eq:class-nlcd-sub2}
\omega_c(NLC_d^{s}) \leq \frac{1}{d^n} \left[ (d-m_t+1) + \Lambda_{NLC_d} (2 m_t - d - 1) \right].
\end{equation}
Comparing Eqs.(\ref{eq:class-nlcd-sub2}) and (\ref{eq:class-nlcd-sub}), we see that with $\Lambda_{NLC_d} \geq \frac{1}{2}$ and $1 \leq m_t \leq d$, $\omega_c(NLC_d^{s}) = \frac{1}{d^n}\left( 1 + (d-1) \Lambda_{NLC_d} \right)$. The fact that the optimal classical strategy for $NLC_d$ is also optimal for the subgames $NLC_d^{s}$ then implies by Lemma \ref{facet-lem} that the Bell inequalities corresponding to the games $NLC_d$ are not facets of the Bell polytope. Again, note that the inequalities $NLC_d^{s}$ themselves may be violated in quantum theory.   
\end{proof}

\section{Faces in the (2,2,2,2,2) CHSH polytope}
\label{sec:CHSH-face}
We have so far seen that many of the paradigmatic Bell inequalities with no quantum violation only describe low-dimensional faces of the corresponding Bell and correlation polytopes. In this section, we look at the simplest and most well-studied Bell scenario, namely that of two parties each measuring two binary observables; explicitly Alice performs binary measurements $A_1, A_2$ and Bob measures $B_1, B_2$. The Bell polytope in this scenario is an $8$-dimensional polytope with the only non-trivial facets known to be the CHSH inequalities (up to local relabelings of inputs and outputs and exchange of parties), and the correlation polytope in this scenario is a $4$-dimensional polytope also with the tight inequalities being the CHSH inequalities. In the set of quantum boxes $\mathcal{Q}$, and the set $\mathcal{E}$ of quantum correlations $\langle A_i B_j \rangle$, these inequalities are well-known to be violated up to the Tsirelson bound. Here, we identify a necessary and sufficient condition for a correlation inequality in this scenario to be not violated in quantum theory, as such we identify the faces of the set of quantum correlations in this simplest scenario. 
Recall that we are not interested in the trivial facets $\langle A_i B_j \rangle = \pm 1$ that are also facets of the no-signaling polytope, we are looking for inequalities that have $\omega_c = \omega_q < \omega_{ns}$. 

Any correlation inequality here can be written up to relabelings in the form
\begin{eqnarray}
\label{eq:chsh-corr-1}
&&p_1 \langle A_1 B_1 \rangle + p_2 \langle A_1 B_2 \rangle + p_3 \langle A_2 B_1 \rangle - p_4 \langle A_2 B_2 \rangle 
\nonumber \\
&&\qquad \leq 1 - 2 \min\{p_i \},
\end{eqnarray}
where we normalize to $\sum_{i} p_i = 1$ and we may choose without loss of generality 
$\min\{p_i \} = p_4$. Equivalently, this is the XOR game
\begin{eqnarray}
\label{eq:chsh-corr}
&&p_1 P(a \oplus b = 0 | A_1, B_1) + p_2 P(a \oplus b = 0|A_1, B_2) +\nonumber \\&& p_3 P(a \oplus b = 0 | A_2, B_1) + p_4 P(a \oplus b = 1|A_2, B_2) \nonumber \\ && \qquad \leq 1 - p_4.
\end{eqnarray} 

\begin{prop}
The necessary and sufficient condition for the weighted CHSH inequality (\ref{eq:chsh-corr-1}) to describe a non-trivial face of the set of quantum correlations is that $p_4 < p_1,p_2, p_3$ and
\begin{equation}
\label{eq:chsh-face}
(p_2 p_3 + p_1 p_4)^2 \leq (p_1 + p_2)(p_1 + p_3)(p_2 - p_4)(p_3 - p_4).
\end{equation}
\end{prop}
\begin{proof}
In \cite{RKM+14}, we identified a necessary and sufficient condition for an XOR game to have no quantum advantage. For the game matrix given as $\tilde{\Phi} = \begin{pmatrix}
  p_1 & p_2  \\
  p_3 & -p_4 
 \end{pmatrix}$ 
with $p_4 = \min \{p_i\}$, and for the optimal classical strategy matrix $S_c = | s_A \rangle \langle s_B |$, 
%with $|s_A \rangle = |s_B \rangle = (1,1)^T$ or $|s_A \rangle = - |s_B \rangle = (1,-1)^T$, 
the necessary and sufficient condition for the game to have $\omega_c = \omega_q$ is that $\Sigma, \Lambda \succ 0$ or $\Sigma, \Lambda \prec 0$ and 
 \begin{equation}
 \rho(\Lambda^{-1} \tilde{\Phi}^T \Sigma^{-1} \tilde{\Phi}) = 1. 
 \end{equation}
Here $\Sigma = \text{diag}({\langle i |\tilde{\Phi} |s_B \rangle \langle s_A |i \rangle}_{i=1,2} )$, $\Lambda = \text{diag}({\langle s_A |\tilde{\Phi} |i\rangle \langle i|s_B \rangle}_{i=1,2} )$ and $\rho( \cdot )$ denotes the spectral radius. 

Explicitly, for the weighted CHSH game with $p_4 = \min\{p_i\}$, one of the optimal strategies is for Alice and Bob to output correlated answers $|s_A \rangle = |s_B \rangle = (1,1)^T$, in which case
%with correlated answers $S_c = \begin{pmatrix} 1 & 1 \\ 1 & 1 \end{pmatrix}$  
we have $\Sigma = \begin{pmatrix}
  p_1  + p_2 & 0  \\
  0 & p_3-p_4 
 \end{pmatrix}$ and $\Lambda = \begin{pmatrix}
  p_1 + p_3 & 0  \\
  0 & p_2 -p_4 
 \end{pmatrix}$. 
 %and for the optimal strategy $S_c = \begin{pmatrix} -1 & 1 \\ 1 & -1 \end{pmatrix$, we have $\Sigma = p_2 - p_1 & 0 \\ 0 & p_3 + p_4 so that 
The condition $\Sigma, \Lambda \succ 0$ then gives $p_4 < p_3, p_2$. 
The condition on the spectral radius gives by a straightforward calculation of the eigenvalues of $\Lambda^{-1} \tilde{\Phi}^T \Sigma^{-1} \tilde{\Phi}$ the condition in Eq.(\ref{eq:chsh-face}). This leaves the possibility that either $p_4 = p_1$ or $p_4 < p_1$. When $p_4 = p_1$, the condition Eq.(\ref{eq:chsh-face}) reduces to $p_4 = p_1 = 0$ which are the trivial inequalities with $w_c = 1$.
\end{proof}
Note that analogous conditions hold when $p_4$ is not the minimum, for instance when $\min\{p_i\} = p_2$ we have the equivalent condition $(p_2 p_3 + p_1 p_4)^2 < (p_1 - p_2)(p_1 + p_3)(p_4 - p_2)(p_3 + p_4)$. Also, note that while Eq.(\ref{eq:chsh-corr}) also defines an edge of the $8$-dimensional quantum set $\mathcal{Q}$, there may also exist other inequalities involving local terms for the Bell polytope here that define higher-dimensional faces of $\mathcal{Q}$.

\section{Single-system binary outcome inequalities without quantum violation}
\label{sec:cut-poly}

{The main result we present in this section is tight noncontextuality inequalities that have no quantum violation and cannot be recovered by the CE principle (see definition \ref{defCE}).} Formally, it is stated as follows
\begin{thm}
{There exist tight  nontrivial noncontextuality inequalities (facets of the non-contextual polytope) with no quantum violation (but violated in general non-disturbing theories) that, moreover, are not in the form of a constraint imposed by the Consistent Exclusivity principle. }
\end{thm}
 We focus on dichotomic measurements, where the output labels are taken as $+1$ and $-1$, and contextuality scenarios where the compatibility structure has at most two compatible measurements in each context, and therefore can be represented by a compatibility graph. We will first present two results from which the main theorem follows.

Finding all the facets of the noncontextual polytope for a given scenario can be accomplished by algorithms such as PORTA \cite{PORTA} 
and cdd \cite{cdd} for small scenarios, whereas for more {complex ones the problem becomes computationally intractable. 
Instead of applying these algorithms,} {in this work} {we obtain noncontextuality inequalities by exploiting the connection between the 
\textit{cut polytope} \cite{DL97} of a compatibility graph and its noncontextual polytope. 
Such a connection was first established by Avis \textit{et al} \cite{ito} in Bell scenarios, and here we make the explicit statement 
for general contextuality ones.} {This will allow us to obtain general statements for contextuality by exploring some established results on cut polytopes.}
 
 {Given a graph $G=(V,E)$, a  cut of $G$ is defined by a subset $S \subset V$, and consists of those edges joining a vertex of S to a vertex not in S. The edge incidence vector of a cut is a binary vector $\vec{x}^S$ of size $|E|$. Its components are: $\vec{x}^S(j) = 0$ if edge $j$ is not present in the cut, and $\vec{x}^S(j) = 1$ otherwise.  The cut polytope of G is the convex hull of the edge incidence vectors: $\{ \vec{x}^S | \text{S is a cut} \}$}.
 
We first explore a remark made by Avis \textit{et al} that one can define the cut polytope  in terms of \textit{anti-correlations} 	
\begin{equation}
 \Xs_{\mi \mj}:= p(a\neq b|ij),
 	\end{equation}  	
{where $\Xs_{ij}$ is short for $ 	\vec{x}^S(\{i,j\})$. That is, the FC-NC (full correlation non-contextual) polytope of $G$ and the cut polytope of $G$ are isomorphic, where the bijection between the two sets is given by the affine transformation 
 \begin{equation}
			  \mean{M_iM_j}=1-2x_{ij}.
 \end{equation}
	Now we will show that a similar statement can be made between the NC-polytope of $G$ and the cut polytope of the \textit{suspension graph} of $G$. Recall that the { suspension graph} $\nabla G  := (V',E')$ of a graph G is obtained from G by adding one new vertex $\nvs$ which is adjacent to all other vertices of $G$. 	
 	}

\begin{prop}\label{prop:cut}
Let $G(V,E)$ be a compatibility graph representing a contextuality scenario. {The NC polytope of this scenario is isomorphic to the cut polytope of the suspension graph $\nabla G $, via the affine transformations:
\begin{align}	\label{affine}
		\mean{M_iM_j}=&1-2x_{ij}  \nonumber \\
		\mean{M_i}=& 1-2x_{iO	}.
\end{align}}
\end{prop}

\begin{proof}

{We will show that every edge incidence vector of $\nabla G $ defines a deterministic non-contextual behaviour and vice-versa, from which the statement about the identity of their convex hulls follows. }

{First arrange the components of $\vec{x}^S$ as follows: the first $n$ components relate to the edges of type
$  E_0 \{ \{\nvs, v\} | v \in V \}$, i.e. are $\Xs_{\nvs \, i}$ \,,  and the others to those of the type $E = \{ \{v, w\} | v,w \in V \}$, i.e. are  $\Xs_{ij}$.} 

{On the one hand, to show that the edge incidence vector of a cut S provides a deterministic non contextual 
model for the $\Xs_{ij}$, we only need to prove that the model is noncontextual, that is, that it factorises properly at the level 
of correlators via the map \eqref{affine}. What we mean by this is that we should prove that $\mean{M_i M_j} = \mean{M_i} \mean{M_j}$, i.e. $1 - 2 \, \Xs_{ij} = (1 - 2 \, \Xs_{\nvs\,i}) (1 - 2 \, \Xs_{\nvs\,j})$. For this note that there are only 8 ways of separating the vertices $i$, $j$ and $\nvs$ into groups of `in S’ and `outside S', and a case by case inspection proves the equality.}

{On the other hand, now we need to prove that every deterministic non contextual models is 
mapped to an edge incidence vector of a cut S. That is, for each non contextual deterministic strategy we need to find a cut S. One way to do is it the following: Let S be the set of vertex $\nvs$ plus all the vertices that denote measurements that give outcome $-1$. It is easy to check that the anti-correlations $p(a\neq b|ij)$ coincide with the edge incidence vector of the cut S.}
\end{proof}

Considerable effort has been devoted to characterising the cut polytopes of complete graphs. Many general inequalities are known \cite{DL97}, and we have the complete list of tight inequalities for complete graph with seven or less measurements \cite{k7}. Moreover, it is conjectured that a complete characterisation is known for the case of eight and nine \cite{k89} measurements (all these inequalities can be also found in the website \cite{site}). 

Now we move on to presenting our examples. Consider the case where the compatibility graph is $G=K_n$, i.e. a complete
graph where all the $n$ measurements are pairwise compatible, { but not necessarily jointly compatible}.
When studying quantum correlations, since we are working with projective measurements, {they are indeed jointly measurable}, and therefore these are scenarios where the quantum set coincides with the noncontextual polytope. That is, all the facets of the cut polytope $CUT(\nabla K_n)$ are tight noncontextuality inequalities with no quantum advantage.

{The question of interest now is whether these tight inequalities with no quantum violation are CE inequalities. 
{ In what follows, we will first translate the question in the hypergraph framework of \cite{AFLS} and then characterise the constraints imposed by CE.}

{In K$_n$, the contexts of compatible measurements have two elements, and the quantity of interest is the behaviour $P(ab|\mi \mj)$, where $a,b \in \{-1,1\}$ and $\mi \neq \mj \in \{M_1, \ldots, M_n\}$. These $M_i$ are the measurements in the contextuality scenario, which in this dichotomic case we will also refer to as observables. This allows to interpret the compatibility graph K$_n$ as a hypergraph $H_n$ as follows. In the language of \cite{AFLS}, each context may be considered as a measurement with output set $\{ (-1,-1),(-1,1),(1,-1),(1,1) \}$. Hence, the set of events that characterize the contextuality scenario $H_n$ is $V(H) = \{(ab|\mi \mj) | a,b \in \{-1,1\} \, , \mi, \mj \in \{M_1, \ldots, M_n\} \land \mi<\,j \}$. The hyperedges of $H_n$ correspond to either normalization conditions or no-disturbance conditions. Hence, 
\begin{widetext} 
\begin{equation*}
\begin{aligned}
E = & \big\{ \{(-1,-1|\mi, \mj),(-1,1|\mi, \mj),(1,-1|\mi, \mj),(1,1|\mi, \mj)\} | \mi, \mj \in \{M_1, \ldots, M_n\} \land \mi < \mj \big\} \\
& \cup \big\{ \{(-1,-1|\mi, \mj),(-1,1|\mi, \mj),(1,-1|\mi, \mk),(1,1|\mi, \mk)\} | \mi, \mj, \mk \in \{M_1, \ldots, M_n\} \land \mi \neq \mj \neq \mk \big\}.
\end{aligned}
\end{equation*}
\end{widetext} }

{It is easy to see that for every hyperedge $e \in E$, $\sum_{v \in e} P(v) = 1$, and that the no-disturbance principle is satisfied. Note also that the (family of) hyperedges of the form $\{(-1,-1|\mi,\mj),(-1,1|\mi,\mj),(1,-1|\mi, \mk),(1,1|\mi,\mk)\}$ corresponds (are isomorphic) to the following sequential protocol: first measure $\mi$, if the outcome is $a=-1$ measure then $\mj$, and otherwise $\mk$. The no-disturbance condition assures that these sequential measurements are complete. These type of hyperedges are the so called ``measurement protocols'' in appendix D of \cite{AFLS}.}

As mentioned in section \ref{Preliminaries}, two events are exclusive if and only if they belong to a same hyperedge 
$e \in E$. In what follows, {we characterise the CE inequalities that arise in these scenarios.}

\begin{prop}\label{LOsets} 
{All nontrivial $CE^1$ inequalities are given by: 	
\begin{equation}\label{cecons}
	\begin{aligned}
	-\mean{M_\mi M_\mj} - \mean{M_\mj M_\mk} - \mean{M_\mk M_\mi} &\leq 1 \\
	-\mean{M_\mi M_\mj} + \mean{M_\mj M_\mk} + \mean{M_\mk M_\mi} &\leq 1 \\
	\mean{M_\mi M_\mj} - \mean{M_\mj M_\mk} + \mean{M_\mk M_\mi} &\leq 1 \\
	\mean{M_\mi M_\mj} + \mean{M_\mj M_\mk} - \mean{M_\mk M_\mi} &\leq 1 
	\end{aligned},
	\end{equation}}
		for every choice of three measurements $M_i,M_j,M_k$. 
\end{prop}

\begin{proof}

The proof consists of two parts. First, we show that the maximum sets of orthogonal events in $H_n$ have size four, and correspond to the trivial constrains given by normalization and no-disturbance conditions. Finally, we show that the non-trivial maximal sets of orthogonal events in $H_n$ have size $3$ and give rise to such CE inequalities.

Let $S$ be a set of mutually orthogonal events on $H_n$. Note that for every two events in $S$, the corresponding contexts must have at least one observable in common. Hence, we divide the proof by cases. 

First, consider the case where all the events in $S$ relate to the same context, which we denote by $\mi\mj$. Hence, the maximal set is $\{ (ab|\mi\mj) : a,b \in \{-1,1\} \}$, where all possible values for the outcomes are allowed.

Second, consider the case where two contexts are encompassed in the events of $S$, i.e. these may have either the form $(\cdot | \mi\mj)$ or $(\cdot | \mi\mk)$. Without loss of generality, suppose $(ab|\mi\mj) \in S$. Then, $(\bar{a}c|\mi\mk) \in S$ for any $c$. Hence, the maximal sets in this case have the form $\{(ab|\mi\mj),(a\bar{b}|\mi\mj),(\bar{a}c|\mi\mk),(\bar{a}\bar{c}|\mi\mk)\}$ for fixed $a,b,c$, i.e. correspond to measurement protocols over K$_n$.

Third, consider the case where three contexts are encompassed in the events of $S$. One option is to take the contexts as $\mi\mj$, $\mi\mk$ and $\mj\mk$. Without loss of generality, suppose $(ab|\mi\mj) \in S$. Then, for an event with $\mi\mk$ to be orthogonal to $(ab|\mi\mj)$, it should have outputs $(\bar{a}c|\mi\mk)$, for any value of $c$. Similarly for $\mj\mk$, we have $(\bar{b}d|\mj\mk)$, for any value of $d$. However, if we want both events to belong to $S$, they should further satisfy $(\bar{a}c|\mi\mk) \perp (\bar{b}d|\mj\mk)$, which implies $d=\bar{c}$. Hence, the maximal sets $S$, characterized by the parameters $a,b,c,\mi,\mj,\mk$, consist of three elements and read
$$
S = \{ (ab|\mi\mj), (\bar{a}c|\mi\mk),  (\bar{b}\bar{c}|\mj\mk)\}.
$$
The other option is to take the contexts as $\mi\mj$, $\mi\mk$ and $\mi\ml$, i.e. they all have $\mi$ as the common observable. It's easy to check that, since the observables are binary, there is no possible assignment of outputs for which the events are pairwise orthogonal. Hence, this choice of three contexts does not give rise to any set of mutually exclusive events. 

Finally, we are left with the case where four or more contexts are encompassed in the events of $S$. In this case, since every pair of contexts should have one observable in common, the only possibility is that all of them share the same observable, i.e. $\mi \mj_1,\ldots, \mi \mj_k$. However, since these are binary observables, it is not possible to assign outputs in a way that assures the pairwise orthogonality of every pair of events with different contexts. Hence, such a situation may never give rise to a set of mutually exclusive events. 

Now we move on to unravel the constraints that CE$^1$ imposes over the space of events in $H_n$. For the first case, those are $\sum_{a,b} P(ab|\mi\mj) \leq 1$, which is trivially satisfied due to normalization of the behaviour $P$. For the second case, the constraints take the form $\sum_b P(ab|\mi\mj) + \sum_c P(\bar{a}c|\mi\mk) \leq 1$, which is guaranteed by the no-disturbance condition.
Indeed, the no-disturbance condition implies $\sum_c P(\bar{a}c|\mi\mk) = \sum_c P(\bar{a}c|\mi\mj)$, hence the statement follows from normalization. Finally, in the third case the CE$^1$ constraint reads $P(ab|\mi\mj)+P(\bar{a}c|\mi\mk)+P(\bar{b}\bar{c}|\mj\mk) \leq 1$, which is a nontrivial CE inequality.

All the possible choices of the outcomes $a,b,c$ will give rise to all the CE inequalities for each choice of three measurements $M_i,M_j,M_k$. When translated into the correlators $\mean{M_\mi M_\mj}$ they take the form of eq.~ \eqref{cecons}. 

\end{proof}
{Note that this lemma is a peculiarity of contextuality scenarios which do not arise from Bell Scenarios. Indeed, due to the local structure of the latter, the orthogonality between $(ab|\mi\mj)$ and $(\bar{b}\bar{c}|\mj\mk)$ is not defined for $\mi \neq \mj  \neq \mk$. Hence, there is no contradiction with the fact that LO$^1\equiv$ NS for bipartite Bell Scenarios \cite{LO}. }

{We have seen that the inequalities that define the set of behaviors that satisfy CE$^1$ have a very simple form. We remark that the inequalities presented at eq \eqref{cecons} are also known as \textit{triangle inequalities} in the cut polytope literature (see chapter 27 of \cite{DL97}) and they were already explored in noncontextuality \cite{triangle}. Also, they represent all tight NC-inequalities in the the complete compatibility graph scenario with three vertices \cite{liang2,AQB+13}.}
	
	 However, the set of noncontextual behavious in $K_n$ (which coincides with the quantum set) is defined by the facets of the CUT polytope of $\nabla K_n$ (see  \cite{site} for a list of all known cut polytope inequalities for complete graph scenarios with less than ten measurements ). Therefore, facet defining inequalities for $CUT(\nabla K_n)$ which are not of triangular form are tight noncontextuality inequalities with no quantum violation and not in the form of a Consistency Exclusivity constraint.

One example of this is the \textit{pentagonal inequality} 
	\begin{equation}
	\sum_{1\leq i<j<n} b_i b_j x_{ij} \leq 0,% \; \text{with } b_1=b_2=b_3=1, b_4=b_5=-1 \; 
	\end{equation}
	with $b_1=b_2=b_3=1, b_4=b_5=-1$, which is facet defining for the cut polytope of a complete graph with five vertices.
	%\MT{do you think it's better to write all the 10 coeficients instead of the sum?}
	Let us associate vertex $5$ with the additional vertex $\nvs$ in $\nabla K_4$. A direct application of the affine map \eqref{affine} provides us the tight noncontextuality inequality	
		\begin{equation}\label{ineq:pent}
		\sum_{1\leq i<j<n} -b_i b_j \mean{M_iM_j} \leq 2,  \; 
		\end{equation}
with $b_1=b_2=b_3=1, b_4=b_5=-1$, where $M_5 = \openone$. Direct inspection shows that this inequality is not in the form of CE, and therefore provides one explicit example of a tight noncontextuality inequality without quantum violation we are searching for. 

In addition to the database presented in \cite{site}, chapters 28-30 of \cite{DL97} present a family of inequalities for the cut polytope of any complete graph.
%In addition to the database presented in \cite{site},  chapters 28-30 of \cite{DL97} present infinitely many classes (inequivalent via relabelling and lifting) of inequalities for the cut polytope of complete graph scenarios. 
One class is given by the \textit{hypermetric inequalties} that are respected by any point in the cut polytope of $K_n$
	\begin{equation}
			\sum_{1\leq i<j<n} b_i b_j x_{ij} \leq 0, \; \text{where } b_i\in \mathbb{Z} \text{ and } \; \sum_i b_i=1.
		\end{equation}
Triangle inequalities and the pentagonal one are particular examples of hypermetric inequalities. 
For the hypermetric inequalities to be tight conditions need to be imposed on the coeffients $b_i$; some of these are known (see chapter 28 of \cite{DL97}). For instance, the inequalities are tight whenever $b_1=b_2=\ldots=b_{n-2}=1,\; b_{n-1}=-1$ and $b_n=4-n$ (theorem 28.2.4 (iiib) of ref \cite{DL97}).
As a remark, note that in general theories that satisfy the no-disturbance condition one can construct boxes where $\mean{M_iM_j} = \pm 1$ while $ \mean{M_i}= 0$, such that they achieve the maximum algebraic value of these inequalities.	 

\section{Conclusion}
In this paper, we presented progress toward the resolution of the question whether there are tight two-party Bell inequalities that are also facets 
of the quantum set of correlations, focusing on the paradigmatic cases of XOR games and linear games. 
We formulated a simple sufficient condition to
identify and construct games with no quantum advantage within these classes. We then showed that the well-known class of Bell inequalities corresponding
to non-local computation tasks are not tight in both the binary outcome case and in the generalization to arbitrary prime outputs for a special class of 
functions, i.e., they do not define facets of the corresponding Bell polytopes. We then identified correlation inequalities that define faces of the set 
of quantum correlations in the simplest Bell scenario corresponding to the CHSH polytope.

Then we moved on to general contextuality scenrios, and studied the question in this single-party scenario, namely non-contextuality inequalities 
with binary outcomes. Here we identified the polytope of non-contextual behaviours with the cut polytope of the suspension graph of the compatibility 
graph representing the measurements in the experiment. Moreover, we characterised the polytope of behaviours that satisfy consistent exclusivity 
inequalities for compatibility graphs that are complete. We then found that, in this particular case (which is not a Bell scenario) there are tight 
non-contextuality inequalities with no quantum advantage that are not in the Consistent Exclusivity form. An open problem is to find an interpretation of 
these inequalities as an information theoretic task. 

Finally, two main questions are still open: whether the two-party correlation polytope shares any facets with the Tsirelson elliptope and whether 
the two-party Bell polytope shares any facets with the quantum correlation set.

\textit{ Acknowledgements.} We are grateful to P. Horodecki and M. Horodecki for useful discussions,
R.R. thanks J.-D. Bancal and Y.-C. Liang for asking the question answered in Section \ref{sec:CHSH-face}. 
R. R. is supported by the ERC AdG grant QOLAPS.
M.T.Q. acknowledges financial support from the Swiss National Science Foundation (grant PP00P2 138917, Starting grant DIAQ).
A.B.S. acknowledges the support from the ERC AdG NLST. 
G. M. acknowledges support from NWO VIDI, ERC Starting Grant and the Brazilian agency Fapemig.
R. A. acknowledges support from the European Union's Horizon 2020 research and innovation programme under the Marie Sk\l{}odowska-Curie 
N-MuQuas Grant No. 705109.


\begin{thebibliography}{99}
\bibitem{footnote1}
Note that Tsirelson's
theorem implies that the set of full-correlation quantum boxes is equivalent to the set of full-correlation boxes described by
first level of the hierarchy introduced in \cite{NPA}.
\bibitem{footnote2}
One can note that the game given by Eq. \eqref{exampleNLC3} 
is described by the function $f(x,y)=x\cdot y+\delta_{x,2}+\delta_{y,2}$,  which is equivalent to $CHSH_3$ up to local terms.
\end{thebibliography}
\end{document}